\documentclass[a4paper]{article}

\usepackage{a4wide}
\usepackage{graphicx,subfigure,hhline}
	\graphicspath{{./}{figures/}}
\usepackage[colorlinks=true,linkcolor=blue,citecolor=red]{hyperref}
\usepackage{xcolor}
\usepackage{amssymb,amsthm,empheq,bbold}
\usepackage[inline]{enumitem}
\usepackage[ruled,vlined,linesnumbered]{algorithm2e}
\usepackage{caption} 
\usepackage{geometry}
\theoremstyle{plain} 
\newtheorem{theorem}{Theorem} 
\newtheorem{prop}{Proposition}
\theoremstyle{remark}

\newcommand{\R}{\mathbb{R}}

\allowdisplaybreaks

\begin{document}
\title{ 
		 An SIR--like kinetic model tracking individuals' viral load
     }
 
 \author{\Large Rossella Della Marca{$^{1}$}, Nadia Loy{$^{2}$}, Andrea Tosin{$^{2*}$}\\[1em]
 	$^1$Risk Analysis and Genomic Epidemiology Unit,\\ Istituto Zooprofilattico Sperimentale della Lombardia e dell'Emilia Romagna,  Parma, Italy\\ rossella.dellamarca@izsler.it 
 	\\[0.5em]
 	$^2$Department of Mathematical Sciences ``G. L. Lagrange'' \\ Politecnico di Torino, Torino, Italy\\nadia.loy@polito.it\\andrea.tosin@polito.it (*corresponding author)}

\date{\today}
\maketitle
\begin{abstract}
Mathematical models are formal and simplified representations of the knowledge related to a phenomenon. 
In classical epidemic models, a neglected aspect is the heterogeneity of disease transmission and
progression linked to the viral load of each infectious individual.
Here, we  attempt to investigate the interplay between the evolution of individuals' viral load and the epidemic dynamics from a theoretical point of view.
In the framework of multi--agents systems, we propose a particle stochastic model describing the infection transmission  trough interactions among agents and the individual physiological course of the disease. 
Agents have a double microscopic state: a discrete label, that denotes the epidemiological compartment to which they belong and switches in consequence of a Markovian process, and a microscopic trait, representing a normalized measure of their viral load, that changes in consequence of binary interactions or interactions with a background. Specifically, we consider Susceptible--Infected--Removed--like dynamics where infectious individuals may be isolated from the general population and the isolation rate may depend on the viral load sensitivity and frequency of tests. 
We derive kinetic evolution equations for the distribution functions of the viral load of the individuals in each compartment, whence, via suitable upscaling procedures, we obtain a macroscopic model for the densities and viral load momentum. We perform then a qualitative analysis of the ensuing macroscopic model, and we present numerical tests in the case of both constant and viral load--dependent isolation control. Also, the matching between the aggregate trends obtained from the macroscopic descriptions and the original particle dynamics simulated by a Monte Carlo approach is investigated.
\end{abstract}
\smallskip

\noindent{\bf Keywords:} Boltzmann--type equations, Markov--type jump processes, epidemic, SIR model, basic reproduction number, viral load, qualitative analysis

\smallskip

\noindent{\bf Mathematics Subject Classification:} 35Q20, 35Q70, 35Q84, 37N25

\section{Introduction}
Mathematical models of infectious diseases spreading have played a significant role in
infection control. On the one hand, they have given an important contribution to the
biological and epidemiological understanding of disease outbreak patterns; on the other hand,
they have helped to determine how and when to apply control measures in order to quickly
and most effectively contain epidemics \cite{Brauer}. Research in this field is constantly
evolving and ever new challenges are launched from the real world (just think of the ongoing
COVID--19 pandemic). One among the many increasingly attractive topics is the mutual influence between the individual behaviours and choices and the disease dynamics \cite{BBRDMCovid,spva}.

In mathematical epidemiology literature a prominent position is occupied by the \textit{compartmental} epidemic models. They are macroscopic 
models where  the total population is divided into disjoint \textit{compartments} according to the individual status with respect to the disease, and the switches from a compartment to another follow given  transition rules. The size of each compartment represents a state variable of the model, whose rate of change is ruled by a \textit{balance} differential equation.
The milestone of compartmental models is the  well--known deterministic Susceptible--Infected--Removed  (SIR) model, proposed by  Kermack and McKendrick  in 1927 \cite{SIR}. 

Like any mathematical model, also epidemic models postulate some simplifying assumptions that are needed to make them analytically tractable and/or numerically solvable. Quantifying the impact of such simplifications is extremely important to understand the model reliability  and identify its range of application.
For example, deterministic compartmental models are valid for large populations. Hence, they can hardly describe situations in which compartments are almost empty (for example at the onset of an epidemic, when the infectious individuals are very few) and, then, stochastic fluctuations cannot be disregarded.

A significant aspect neglected by classical epidemic models is the heterogeneity of disease transmission and progression linked to the \textit{viral load} of each infectious individual. 
Viral load is defined as a quantitative viral titre (e.g. copy number) \cite{ecdc} and may represent a useful marker for assessing viral kinetics, disease severity and prognosis. Indeed,   symptoms and mortality induced by the infection may depend on the  individual viral load, like asserted, for example, by studies on  seasonal flu  \cite{flu}, measles \cite{Simmonds} and  COVID--19 disease \cite{Fajnzylber2020}.  
 The quantity of virus in the organism can also influence the results by screening and diagnostic tests, which  are capable of detecting a different quantity of virus per $ml$ according to their sensitivity. Hence, the viral load affects the probability for an individual of being diagnosed and, consequently, home isolated or hospitalized, thus preventing the possibility of him/her infecting other people. In this context, assessing the  interplay between the frequency of testing and sensitivity of the tests \cite{larremore2021SA} is crucial for planning  prevention and mitigation measures.
  Also the timing of testing is fundamental: for example in the case of acute rubella, in order to have laboratory confirmation of infection, viral specimen should be collected as soon after symptom onset as possible, preferably one to three days after onset, but no later than seven days post--onset \cite{cdcRubella}.
Last but not least, 
the viral load can be a strong determinant of transmission risk \cite{Goyal2020} and the knowledge of the duration of viral shedding plays a key role in tracing the evolution of the infectious disease 
\cite{cevik2020virology}. 
For example, it is estimated that SARS--CoV--2 viral load  peaks just before the symptom onset, i.e. during the \textit{pre--symptomatic} stage of infection \cite{He2020,ecdc} and  pre--symptomatic patients are responsible of about $44\%$ of secondary infections  \cite{He2020}. In the case of congenital rubella syndrome, infants can shed the virus up to one year, but samples should be collected prior to three months of age because by three months of age approximately 50$\%$ will no longer shed virus \cite{cdcRubella}.  
 
The mathematical framework of multi--agent systems \cite{pareschi2013BOOK} allows one to introduce a detailed microscopic description of the interactions between individuals,
that are generally called \textit{agents}, within a population. One of the key aspects is that it allows one to recover a statistical description of the system by introducing a probability density function accounting for the statistical distribution of some microscopic traits of the individuals. Its evolution may be described by \textit{kinetic equations} that also permit to derive hydrodynamic equations, i.e. macroscopic models, that, thus, inherit a large number of features of the original microscopic dynamics. In particular, the authors in \cite{loy2020KRM} introduced a particle model describing a microscopic dynamics in which agents have a double microscopic state: a discrete label that switches as a consequence of a Markovian process and a microscopic trait that changes as a consequence of \textit{binary interactions} or \textit{interactions with a background}. The trait may take into account the individual viral load, while the label denotes the compartment to which the agent belongs. The authors then derived nonconservative kinetic equations describing the evolution of the distribution of the microscopic trait for each label and, eventually, hydrodynamic equations for the densities and momentum of the microscopic trait of each compartment.

Kinetic equations have been applied to compartmental epidemic models in order to take into account the role of wealth distribution during the spread of infectious diseases, for example in \cite{dimarco2020PRE, dimarco2021PREPRINT}.  In these works, the authors described in more detail social contacts among the individuals but still relied on a SIR--like structure to model contagion dynamics. To our best knowledge, the only kinetic model taking into account the spread of an infectious disease by  means of interactions and including the individuals' viral load is the one proposed in \cite{loyTosin2021PREPRINT}, where, however, the authors did not consider epidemiological compartments.

Motivated by the previous arguments, in the present work we  propose a microscopic stochastic model allowing one to describe the spread of an infectious disease as a consequence of the interactions among individuals who are characterized by means of their viral load. Once infected, the viral load of the individuals increases up to a maximum peak and then decreases as a consequence of a physiological process. Furthermore, the individuals are labelled in order to indicate their belonging to one of the disjoint epidemiological compartments. Specifically, we consider an SIR--like dynamics with an isolation mechanism that depends on the individual viral load (Section \ref{sec:2}). We  derive kinetic evolution equations for the distribution functions of the viral load of the individuals in each compartment and, eventually, a macroscopic model for the densities and viral load momentum  (Section \ref{sec:3}). We  perform then a qualitative analysis of the ensuing macroscopic model (Section \ref{sec:4}), and we present some numerical tests of both the microscopic and macroscopic models  to show the matching between the aggregate trends obtained from
the macroscopic descriptions and the original particle dynamics simulated by a Monte Carlo approach (Section \ref{sec:5}). Finally, we draw some conclusions and we briefly sketch possible research developments (Section \ref{sec:6}).

\section{ A multi--agent system describing the spread of an epidemics through interactions}\label{sec:2}

Let us consider a large system of interacting individuals in presence of an infectious disease that spreads  through social contacts. 
The total population at time $t$ is  divided into disjoint epidemiological compartments according to the health status with respect to the disease, to each of which we associate a label $x \in \mathcal{X}$.
 Individuals, that we shall also call the \textit{agents}, are characterized  by the  evolution stage of the disease--related viral load, that is the number of viral particles present in the organism. Let us denote with  $v\in[0,1]$ a normalized measure of the individual viral load at time $t$, where $v=1$ represents the  maximum observable value.
We want to describe the microscopic mechanisms modelling the interactions between individuals,  which are the means of  the transmission of the disease, and the switch of compartment of each individual that follows from the disease progression.

Being our final aim  the proposal of a macroscopic model, we  derive as an intermediate stage a statistical description of our multi--agent system through kinetic equations, by which we  then recover hydrodynamic limits. In order to give a statistical description of the multi--agent system, whose total mass is conserved in time, we introduce a distribution function for describing the statistical distribution of the agents characterized by the pair $(x,v) \in \mathcal{X}\times[0,1]$, as
\begin{equation}\label{eq:f.delta.label_switch_net}
f(t,x, v)=\sum_{i=1}^n \delta(x-i)  f_{i}(t,v),
\end{equation}
where $\delta(x-i) $  is the Dirac delta distribution centred at $x=i$,
and we  assume that $f(t,x,v)$ is a probability distribution, namely:
\begin{equation}
	\int_{0}^1\int_\mathcal{X} f(t,x,v)dxdv=\sum_{i=1}^{n}\int_{0}^1 f_i(t,v)dv=1, \quad \forall\,t\ge0. 
	\label{eq:f.prob}
\end{equation}
In (\ref{eq:f.delta.label_switch_net})--(\ref{eq:f.prob}), $f_i=f_i(t,v)\geq 0$ is the distribution function of the microscopic state $v$ of the agents that are in the $i$--th compartment at time $t$. Hence, $f_i(t,v)dv$ is the proportion of agents in the compartment $i$, whose microscopic state lies between $v$ and $v+dv$ at time $t$.
In general, the  $f_i$'s, $i \in \mathcal{X}$, are  not probability density functions because their $v$--integral varies in time due to the fact that agents move from one compartment to another. We denote by
\begin{equation*}
	\rho_i(t):=\int_{0}^1 f_i(t,v)dv
\end{equation*}
the density of agents in the class $i$, thus $0\leq\rho_i(t)\leq 1$ and
$$ \sum_{i=1}^{n}\rho_i(t)=1, \quad \forall\,t\ge 0. $$
Then, we define the \textit{viral load momentum} of the $i$th compartment as the first moment of $f_i$ for each class $i \in \mathcal{X}$, i.e. 
\begin{equation*}
n_i(t)=\int_{0}^1 f_i(t,v)vdv.
\end{equation*}
If $\rho_i(t)>0$ we can also define the \textit{mean viral load} as the ratio
${n_i(t)}/{\rho_i(t)}.
$
We observe that $\rho_i(t)=0$ implies instead necessarily  $f_i(t,v)=0$ and, therefore, $n_i(t)=0$. In such a case, the mean viral load is not defined because the corresponding  compartment is empty.

\subsection{The compartmental structure}
The individuals, labelled with $x\in \mathcal{X}$, are divided in the following disjoint epidemiological  compartments:
\begin{itemize}
\item \textit{susceptible}, $x=S$:  individuals who are healthy but can contract the disease. The susceptible population
increases by a net inflow, incorporating  new births and immigration, and decreases due to disease transmission and natural death;
\item \textit{infectious}, $I$: individuals who  are infected by the disease and can transmit the virus to others. We assume that members of this class are asymptomatic or mildly symptomatic, hence they move freely. Infectious individuals arise as the result of new infections of susceptible individuals and diminish due to recovery  and natural death or because they are identified and isolated from the general population;
\item\textit{isolated}, $H$:  infected individuals who have been identified and fully isolated from the general population by home isolation or hospitalization. Members of this class come from the infectious compartment $I$ and get out due to recovery or death. We assume that this class includes patients showing severe symptoms, that can also die due to the disease; 
\item  \textit{recovered}, $x=R$:  individuals who are recovered from the disease after the infectious period. They come from the infected compartments $I$ and $H$
and acquire long lasting immunity against the disease.
\end{itemize}
 Specifically: susceptible individuals have $v\equiv0$; once infected, an individual  viral load increases until reaching a peak value (that varies from person to person) and then gradually decreases, see e.g. the  representative plot of SARS--CoV--2 viral load evolution given in \cite{cevik2020virology}, Fig. 2. Hence, for mathematical convenience, we assume that members of classes $I$ and $H$ are further divided into:
\begin{itemize}
	\item  infectious, $x=I_1$, and isolated, $x=H_1$,  with increasing viral load;
	\item infectious, $x=I_2$, and isolated, $x=H_2$,  with decreasing viral load.
\end{itemize}
Note that new infections enter the class $I_1$, while recovery may occur only during the stages $I_2$ or $H_2$. Finally, after the infectious period, recovered individuals may still have a positive viral load  which however definitively  approaches  zero.

Also, since our model incorporates \textit{birth} and \textit{death} processes, we  introduce the following three auxiliary compartments:  individuals that enter the susceptible class by newborn or immigration, $x=B$; individuals who die of natural causes, $x=D_\mu$; and individuals who die from the disease, $x=D_d$.  We assume that members of class $B$ have $v\equiv 0$, while those of classes $D_\mu$ and $D_d$ retain the viral load value at time they died.

\subsection{Evolution of the viral  load}\label{Sec:viral-load}
Let us now focus on the mathematical modelling of the evolution of an individual viral  load $v$. 
We distinguish  the two following cases when   $v$ changes over time: i) a susceptible individual,  having $v=0$, acquires a  positive viral load (and get infected) by interaction with an infectious individual; ii)  the viral  loads of infected ($I_1$, $I_2$, $H_1$, $H_2$) and recovered ($R$) individuals evolve naturally in virtue of physiological processes. 

Given an agent labelled with $S$, then the necessary condition for acquiring positive viral load is an encounter with an infectious individual ($I_1$ or $I_2$).  Let us denote with $\lambda_\beta>0$ the frequency of these interactions. Increasing [resp. decreasing] $\lambda_\beta$ corresponds to increasing [resp. reducing] encounters among people: the lower $\lambda_\beta$  the more strengthened social distancing.  

By interacting with an infectious individual, a susceptible individual may or may not get infected. In the first case his/her  viral load after the interaction (say, $v'$) is positive: $v'>0$; in the second case it remains null: $v'=0$. Specifically, we consider the following microscopic rule:  
\begin{equation*}
v'=X_{\nu_\beta}v_0,
\end{equation*}
where $X_{\nu_\beta}$ is a Bernoulli random variable of parameter $\nu_\beta\in (0,1)$  describing the case of successful contagion when $X_{\nu_\beta}=1$ and the case of contact without contagion when $X_{\nu_\beta}=0$. 
We assume that new infected individuals enter the  class $I_1$ and they all acquire the same initial viral load, $v_0$ (that can be interpreted as an average initial value). 
We remark that this binary interaction process causes simultaneously a change of the microscopic state $v$ and a label switch, because as soon as $v$ becomes positive, i.e. if $X_{\nu_\beta}=1$, the susceptible individual  switches to the class $I_1$. 

Infectious, isolated and recovered individuals cannot change their viral  load in binary interactions, but the evolution reflects  physiological processes. In particular, starting from the initial positive value $v=v_0$, the viral load increases until reaching a given peak value and then it decreases towards zero. The peak can be reached in either the stage $I_1$ or $H_1$, i.e. before or after an infectious individual is possibly isolated.

In this framework, the microscopic state $v$  varies as a consequence of an autonomous process (also called \textit{interaction with a fixed background} in the jargon of multi--agent systems \cite{pareschi2013BOOK}). 
Specifically,
given an agent $(I_1,v)$ or $(H_1,v)$, namely an infected individual with increasing viral load, we consider a linear--affine expression for the microscopic rule describing the evolution of $v$ into a new viral load $v'$:
\begin{equation}\label{eq:micro_pre.max}
v'=v+\nu_1(1-v). 
\end{equation}
The latter is a prototype rule describing the fact that the viral load may increase up to a certain threshold normalized to 1 by a factor proportional to $(1-v)$. In particular, $\nu_1\in (0,1)$ is the factor of increase of the viral load.

Similarly, given an agent $(I_2,v)$, $(H_2,v)$ or $(R,v)$, namely an infected individual with decreasing viral load or a recovered individual, we consider the following microscopic rule for the evolution of $v$:
\begin{equation}\label{eq:micro_post.max}
v'=v-\nu_2 v,
\end{equation}
being the parameter $\nu_2\in(0,1)$ the factor of decay of the viral  load. These microscopic processes happen with frequency $\lambda_\gamma>0$.
We observe here that the introduction of the sub--classes $I_1,\,I_2$ and $H_1,\,H_2$ is needed in order to implement the microscopic rules  \eqref{eq:micro_pre.max}--\eqref{eq:micro_post.max} in a kinetic equation. 
These two rules are deliberately generic and very simple: the only aim is to distinguish individuals based on whether their viral load is increasing or decreasing and to implement two different factors $\nu_1$ and $\nu_2$ accordingly.

\subsection{A microscopic stochastic model}\label{Sec:micro}
Let us now define a microscopic stochastic process implementing the modelling assumptions defined so far. Let us  consider an agent  characterized by the pair of random variables $(X_t,V_t)$, where $X_t \in \mathcal{X}$ is the label denoting the compartment to which the agent belongs and $V_t\in [0,1]$ is the  viral load. Then, the random variable $X_t$ changes in consequence of a Markovian jump process, while the microscopic state $V_t$ may change either because of a binary interaction with an agent characterized by $(Y_t,W_t)$, $Y_t \in \mathcal{X}$, $W_t \in [0,1]$ or because of an autonomous process according to the discussion in Section \ref{Sec:viral-load}. In particular, two different types of stochastic processes may happen:
\begin{enumerate}
\item the agents may change their label according to a process that is independent of the change of the viral load: birth and death processes, $(I_1\rightarrow H_1)$, $(I_2\rightarrow H_2)$, where the notation $(j\rightarrow i)$ indicates the switch from compartment $j$ to compartment $i$;
\item the agents may change their viral load and their label simultaneously: $(S\rightarrow I_1)$, $(I_1\rightarrow I_2)$, $(I_2\rightarrow R)$, $(H_1\rightarrow H_2)$, $(H_2\rightarrow R)$. This class of processes also includes the evolution of the viral load of individuals who remain in the same compartment, i.e. $(i\rightarrow i)$, $\forall i \in \mathcal{X}$. 
\end{enumerate} 
 The two stochastic processes above may be expressed in the following rule describing the variation of $X_t$ and $V_t$ of a generic representative agent of the system during a time interval $\Delta t>0$:
\begin{equation}\label{stoch_proc}
\begin{aligned}
(X_{t+\Delta t},V_{t+\Delta t})=& \Sigma\left[(1-\Theta)(X_t,V_t)+\Theta (J^v_{X_t},V^{'}_{X_t})\right]+\Psi\Big[ ((1-\Xi)X_t,V_t)+(\Xi J_{X_t},V_{t})\Big],
\end{aligned} 
\end{equation}
where $\Sigma$ and $\Psi$ are indicator functions. In particular, $\Sigma=1$ if the agents of compartment labelled with $X_t$ change their viral load and label simultaneously ($\Sigma=1$ if $X_t \in \lbrace I_1,I_2\rbrace$) and $\Psi=1$ if the label of the agents in compartment $X_t$ changes independently of the viral load (if $X_t \in \lbrace S,I_1,I_2, H_1, H_2,R\rbrace$). $J_{X_t}$ is the new label of an agent performing a label switch independently on the viral load and previously labelled with $X_t$. Moreover, $V^{'}_{X_t}$, $J^v_{X_t}$ are the new viral load and label of an agent with previous state $(X_t,V_t)$ in the case of a process in which the viral load and the label change simultaneously. Furthermore, we assume that $\Theta$ and $\Xi$ are two independent Bernoulli random variables describing whether a process happens ($\Theta=1,\,\Xi=1$) or not ($\Theta=0,\,\Xi=0$). We suppose that $P(\Theta=1)=\lambda_{X_t} \Delta t$,  being $\lambda_{X_t}$ the frequency of the microscopic process that rules the change of the microscopic variable $v$, while $P(\Xi=1)=\lambda_{X_t,J_{X_t}}\Delta t$, where $\lambda_{X_t,J_{X_t}}$ is the frequency of the transition that causes the independent label switch from $X_t$ to $J_{X_t}$. In order for $P$ to be well defined we must have that $\lambda_{X_t}\Delta t,\, \lambda_{X_t,J_{X_t}}\Delta t \le1$. The latter models the assumption according to which the larger the time interval, the higher the
probability of having a label switch. 

\paragraph{Independent label switch}
Let us denote with $P(j\rightarrow i):=P(J_{X_t}=i|X_t=j)$ the conditional probability of switching from compartment $j$ to compartment $i$, with $i,\,j\in\mathcal{X}$, independently of a change of the microscopic state $v$. This probability concerns birth and death processes, and the isolation of infectious individuals, i.e. the label switches $(I_1\rightarrow H_1)$, $(I_2\rightarrow H_2)$. Specifically, we consider the following non--zero values for $P$:
\begin{itemize}
\item $P(B\rightarrow S)={b}/{\rho_B(t)}\in[0,1]$, where $b$ is a non--negative constant; 
\item $P(S \rightarrow D_\mu)=P(I_1\rightarrow D_\mu)=P(I_2\rightarrow D_\mu)=P(H_1\rightarrow D_\mu)=P(H_2\rightarrow D_\mu)=\mu\in[0,1]$;
\item $P(H_1 \rightarrow D_d)=P(H_2 \rightarrow D_d)=d\in[0,1]$; 
\item $P(I_1 \rightarrow H_1)=P(I_2 \rightarrow H_2)=\alpha_H(v) \in [0,1]$, where $\alpha_H(v)$ is an increasing function of the viral  load $v$. It accounts for the fact that infectious people with higher viral load are more likely to be identified. Indeed,  performances of screening and diagnostic tests increase with the actual number of viral particles in the organism (see e.g. the interim guidance \cite{whotest} on diagnostic testing for SARS--CoV--2).  Moreover, for some infectious diseases a higher viral load is positively associated with a worse outcome and symptomatology (like for seasonal flu \cite{flu}).
\end{itemize} 
The frequencies of the Markovian processes describing the switch between the different compartments may, in general, depend on both the departure and the arrival classes. It means that the process of switching from class $j$ to class $i$, that happens with probability $P(J_{X_t}=i|X_t=j)$ happens with a frequency $\lambda_{X_t,J_{X_t}}=\lambda_{i,j}$. In particular, we consider
\begin{itemize}
\item $\lambda_{S,B}=\lambda_b$, that is the frequency of new births or immigration;
\item $\lambda_{D_\mu,j}=\lambda_\mu, \,  j\in \mathcal{X}\setminus \{B,D_\mu,D_d\}$, that is the frequency of natural deaths;
\item $\lambda_{D_d,H_1}=\lambda_{D_d,H_2}=\lambda_d$, that is the frequency of disease--induced deaths;
\item $\lambda_{H_1,I_1}$, $\lambda_{H_2,I_2}$,  that are the frequencies at which infectious individuals are isolated. 
\end{itemize}

\paragraph{Simultaneous label switch}
In our multi--agent system  the first microscopic process causing simultaneously a label switch and a progression of the viral load is the transition from susceptible ($S$) to infectious ($I_1$)  state. This process has frequency $\lambda_{\beta}>0$.
We express the corresponding transition probability as
$$
P(J_{X_t}^v=I_1,V'_{X_t}=v'|X_t=S)
$$
that is the probability for an agent labelled with $X_t=S$ to change his/her label and zero viral load into $(I_1,v')$. Since this happens if a susceptible individual meets an infectious individual,  we may regard $P(J_{X_t}^v=I_1,V'_{X_t}=v'|X_t=S)$ as a probability density distribution of the joint random variables $J_{X_t}^v$ and $V^{'}_{X_t}$, given the probability $\rho_{I_1}(t)+\rho_{I_2}(t)$ of encountering an infectious individual, i.e.
\[
P(J_{X_t}^v=I_1,V'_{X_t}=v'|X_t=S)=P(J^v_{X_t}=I_1,V'_{X_t}=v')(\rho_{I_1}(t)+\rho_{I_2}(t)),
\]
that can be rewritten as
\[
P(J_{X_t}^v=I_1,V'_{X_t}=v'|X_t=S)=P(J^v_{X_t}=I_1|V'_{X_t}=v')P(V^{'}_{X_t}=v')(\rho_{I_1}(t)+\rho_{I_2}(t)),
\]
where $P(J^v_{X_t}=I_1|V'_{X_t}=v')=1$ if $v>0$, and $P(J^v_{X_t}=I_1|V'_{X_t}=v')=0$ if $v=0$. $P(V^{'}_{X_t}=v')$ is the probability density distribution of the random variable $V^{'}_{X_t}=X_{\nu_\beta} v_0$  and it takes into account the microscopic rule describing the change of the state $v$ in terms of transition probabilities (see \cite{loy2020CMS} for more details). It may be also expressed as $P(V^{'}_{X_t}=v')=\nu_\beta P_S(v')$ where $P_S=\delta(v'-v_0)$. 

Analogously, we may express the transition probabilities concerning the autonomous process and label switch as
\[
P(J^v_{X_t}=i,V'_{X_t}=v'|X_t=j,V_t=v)=P(J^v_{X_t}=i|V^{'}_{X_t}=v')P_j(v\rightarrow v'),
\]
where $P(J^v_{X_t}=i|V^{'}_{X_t}=v')$ is the probability density distribution of having an agent labelled with $i$ given that he/she has a viral  load $v'$, while $P_j(v\rightarrow v')$ is the transition probability describing the autonomous process of the viral  load $v'$, given the previous viral  load $v$, for agents labelled with $j$. 

\textbf{Remark.} If $v'$ is such that $P(J^v_{X_t}=I_1|V'_{X_t}=v')=1$ and $i=j$, then $P(J^v_{X_t}=i,V'_{X_t}=v'|X_t=j,V_t=v)=P_j(v\rightarrow v')$ is the transition probability that describes the change of the microscopic state $v'$ alone according to the rules \eqref{eq:micro_pre.max}--\eqref{eq:micro_post.max} (see \cite{loy2020CMS}). 

\medskip

In our case, the transitions  to take into account are:
\begin{itemize}
\item $P(J_{X_t}^v=I_2,V^{'}_{X_t}=v'|X_t=I_1,V_t=v)=P(J_{X_t}^v=I_2|V^{'}_{X_t}=v')P_{I_1}(v\rightarrow v')$ and $P(J_{X_t}^v=H_2,V^{'}_{X_t}=v'|X_t=H_1,V_t=v)=P(J_{X_t}^v=H_2|V^{'}_{X_t}=v')P_{H_1}(v\rightarrow v')$, where $P(J_{X_t}^v=I_2|V^{'}_{X_t}=v')=P(J_{X_t}^v=H_2|V^{'}_{X_t}=v')=\eta(v')$. In principle, the probability $\eta(v')$  should  increase by increasing the viral load $v'$, since individuals with higher viral load are more likely to have reached the peak value.
Here, for mathematical convenience, we approximate $\eta$ to the factor of viral load increase: $\eta=\nu_1\in[0,1]$. 
Both $P_{I_1}(v\rightarrow v')$ and $P_{H_1}(v\rightarrow v')$ have average $v+\nu_1(1-v)$. In particular, we  choose $P_{I_1}(v\rightarrow v')=P_{H_1}(v\rightarrow v')=\delta \big( v'- (v+\nu_1(1-v))\big)$.  
 \item
 $P(J_{X_t}^v=R,V^{'}_{X_t}=v'|X_t=I_2,V_t=v)=P(J_{X_t}^v=R|V^{'}_{X_t}=v')P_{I_2}(v\rightarrow v')$ and $P(J_{X_t}^v=R,V^{'}_{X_t}=v'|X_t=H_2,V_t=v)=P(J_{X_t}^v=R,V^{'}_{X_t}=v')=P(J_{X_t}^v=R|V^{'}_{X_t}=v')P_{H_2}(v\rightarrow v')$, where $P(J_{X_t}^v=R|V^{'}_{X_t}=v')=\gamma(v')$  describes the probability for an infected individual of recovering. In principle, the probability $\gamma(v')$  should  increase by decreasing the viral load $v'$, since individuals with lower viral load are more likely to have passed the infectious period. Similarly to what done for $\eta(v')$, we approximate this probability to the factor of viral load decay: $\gamma=\nu_2\in[0,1]$. 
$P_{I_2}(v\rightarrow v')$ and $P_{H_2}(v\rightarrow v')$ have average $v-\nu_2 v$. In particular, we  choose $P_{I_2}(v\rightarrow v')=P_{H_2}(v\rightarrow v')=\delta\big(v'-v(1-\nu_2)\big)$. 
 
\item $P(J_{X_t}^v=I_1,V^{'}_{X_t}=v'|X_t=I_1,V_t=v)=P_{I_1}(v\rightarrow v')$, $P(J_{X_t}^v=H_1,V^{'}_{X_t}=v'|X_t=H_1,V_t=v)=P_{H_1}(v\rightarrow v')$, while $P(J_{X_t}^v=I_2,V^{'}_{X_t}=v'|X_t=I_2,V_t=v)=P_{I_2}(v\rightarrow v')$, $P(J_{X_t}^v=H_2,V^{'}_{X_t}=v'|X_t=H_2,V_t=v)=P_{H_2}(v\rightarrow v')$.

\item $P(J_{X_t}^v=R,V'_{X_t}=v'|X_t=R,V_t=v)=P_{R}(v\rightarrow v')=\delta\big(v'-v(1-\nu_2)\big)$.
\end{itemize}
The frequency of these transitions is the frequency of the corresponding microscopic process, i.e. $\lambda_{X_t}=\lambda_{\gamma}$ for $X_t \in \lbrace I_1,I_2,H_1,H_2,R \rbrace$. 

\section{Aggregate description: from kinetic to hydrodynamic equations}
\label{sec:3}
The kinetic equations equations describing the evolution of $f_i(t,v),\, i \in \mathcal{X}$, can be derived in the same way as in \cite{LnTa_NonCons}. Namely, 
 the system of the weak equations for the $f_i$'s is the following:
\begin{align}
	\begin{aligned}[b]
		\dfrac{d}{dt}\int_{V}\varphi(v)f_i(t,v)dv =& \int_{V}\varphi(v)\left(\sum_{j=1}^{n}\left[\lambda_{i,j}P(j\rightarrow i)f_j(t,v)-\lambda_{j,i}P(i\rightarrow j)f_i(t,v)\right]\right)dv \\
		 &+  \sum_{j=1}^{n}\sum_{k=1}^n\int_{V} \int_V\left[\lambda_{j}\varphi(v')P(i,v'|j,v)f_j(t,v)-\lambda_{i}\varphi(v)P(j,v\vert i,v')f_i(t,v)\right]dv dv' ,\quad i\in\mathcal{X},
	\end{aligned}
	\label{eq:boltz.fi}
\end{align}
where $\varphi:[0,1]\rightarrow\mathbb{R}$ is a test function.
In \eqref{eq:boltz.fi}, the first and second lines account for the gain term of the Markovian processes describing the label switches that happen, respectively, independently of the evolution of the viral load, and simultaneously with the evolution of the viral load. The frequency $\lambda_{i}$ of the Markovian process due to an interaction with a background corresponds to the frequency of changing the microscopic state $v$ and it is, as previously stated, $\lambda_i=\lambda_\gamma$, $\forall i \in \mathcal{X}\setminus\lbrace B,D_d,D_\mu \rbrace$.

From (\ref{eq:boltz.fi}), we derive the kinetic equations describing the evolution of the distribution functions $f_i$'s, $i\in\mathcal{X}$.  For $i\in\mathcal{X}\setminus\{B,D_d,D_\mu\}$, namely the classes of living individuals,  we get
\begin{itemize}
\item susceptible individuals  ($i=S$)
\begin{align}
	\begin{aligned}[b]
		\dfrac{d}{dt}\int_{0}^1\varphi(v)f_S(t,v)dv &= \int_{0}^1\varphi(v)\left(\lambda_b \dfrac{b}{\rho_B(t)} f_B(t,v)-\lambda_\mu \mu f_S(t,v)\right)dv \\
		&\phantom{=}-\lambda_{\beta}\nu_\beta \int_{0}^1\int_0^1 \varphi(v') P_S(v') \rho_S(t)\delta(v-0)(\rho_{I_1}(t)+\rho_{I_2}(t))  dv  dv', 
	\end{aligned}
	\label{eq:boltz.fS1}
\end{align}
\item infectious individuals  with increasing viral  load ($i=I_1$) 
\begin{align}
	\begin{aligned}[b]
		\dfrac{d}{dt}\int_{0}^1\varphi(v)f_{I_1}(t,v)dv &=-\int_{0}^1\varphi(v)\left(\lambda_{H_1,I_1}(t)\alpha_H(v)f_{I_1}(t,v) +\lambda_\mu \mu f_{I_1}(t, v)\right) dv \\
		&\phantom{=}+\lambda_{\beta} \nu_\beta\int_{0}^1\int_0^1 \varphi(v')  P_S(v') \rho_S(t)\delta(v-0)(\rho_{I_1}(t)+\rho_{I_2}(t))  dv dv'\\
		&\phantom{=}-\lambda_\gamma\int_{0}^1\int_{0}^1\varphi(v')\eta(v')P_{I_1}(v'|v)f_{I_{1}}(t,v) dv dv'  \\
		&\phantom{=} +\lambda_\gamma\int_{0}^1\int_0^1 (\varphi (v')P_{I_1}(v'|v)f_{I_{1}}(t,v)-\varphi(v)P_{I_1}(v|v')f_{I_{1}}(t,v')) dv  dv', 
	\end{aligned}
	\label{eq:boltz.fI_1}
\end{align}
\item infectious individuals with decreasing viral  load ($i=I_2$)
\begin{align}
	\begin{aligned}[b]
		\dfrac{d}{dt}\int_{0}^1\varphi(v)f_{I_2}(t,v)dv &=- \int_{0}^1\varphi(v)\left(\lambda_{H_2,I_2}(t)\alpha_H(v)f_{I_2}(t,v)+\lambda_\mu \mu f_{I_2}(t, v)\right)dv \\
			&\phantom{=}+\lambda_\gamma\int_{0}^1\int_{0}^1\varphi(v')\eta(v')P_{I_1}(v'|v)f_{I_{1}}(t,v) dv dv' \\
		&\phantom{=}-\lambda_\gamma \int_0^1 \int_0^1\varphi(v') \gamma(v')P_{I_2}(v'|v)f_{I_2}(t, v) dv dv'\\
		&\phantom{=} +\lambda_\gamma\int_{0}^1\int_0^1 (\varphi (v')P_{I_2}(v'|v)f_{I_{2}}(t,v)-\varphi(v)P_{I_2}(v|v')f_{I_{2}}(t,v')) dv  dv', 
		\end{aligned}
	\label{eq:boltz.fI_2}
\end{align}
\item isolated individuals  with increasing viral  load ($i=H_1$)
\begin{align}
	\begin{aligned}[b]
		\dfrac{d}{dt}\int_{0}^1\varphi(v)f_{H_1}(t,v)dv &= \int_{0}^1\varphi(v)\left(\lambda_{H_1,I_1}(t)\alpha_H(v)f_{I_1}(t,v)-\lambda_d d f_{H_1}(t,v)-\lambda_\mu \mu f_{H_1}(t,v)\right)dv \\
		&\phantom{=}-\lambda_\gamma\int_{0}^1\int_{0}^1\varphi(v')\eta(v')P_{H_1}(v'|v)f_{H_1}(t,v) dv dv' \\
		&\phantom{=} +\lambda_\gamma\int_{0}^1\int_0^1(\varphi(v') P_{H_1}(v'|v)f_{H_1}(t,v)-\varphi(v)P_{H_1}(v|v')f_{H_{1}}(t,v')) dv  dv',
	\end{aligned}
	\label{eq:boltz.fH_1}
\end{align}
\item  isolated individuals  with decreasing viral  load ($i=H_2$)
\begin{align}
	\begin{aligned}[b]
		\dfrac{d}{dt}\int_{0}^1\varphi(v)f_{H_2}(t,v)dv &=\int_{0}^1\varphi(v)\left(\lambda_{H_2,I_2}(t)\alpha_H(v) f_{I_2}(t,v)-\lambda_d  d f_{H_2}(t,v)-\lambda_\mu \mu f_{H_2}(t,v)\right)dv \\
			&\phantom{=}+\lambda_\gamma\int_{0}^1\int_{0}^1\varphi(v')\eta(v')P_{H_1}(v'|v)f_{H_1}(t,v)  dv dv \\
			&\phantom{=}-\lambda_\gamma\int_{0}^1\int_{0}^1\varphi(v')\gamma(v')P_{H_2}(v'|v)f_{H_2}(t,v)  dv dv' \\
		&\phantom{=} +\lambda_\gamma\int_{0}^1\int_0^1 (\varphi(v')P_{H_2}(v'|v)f_{H_2}(t,v)-\varphi(v)P_{H_2}(v|v')f_{H_{2}}(t,v')) dv  dv',
	\end{aligned}
	\label{eq:boltz.fH_2_bis}
\end{align}
\item recovered individuals ($i=R$)
\begin{align}
	\begin{aligned}[b]
		\dfrac{d}{dt}\int_{0}^1\varphi(v)f_R(t,v)dv &=-\lambda_\mu \mu \int_{0}^1\varphi(v) f_R(t,v)dv  \\
		&\phantom{=}+\lambda_\gamma \int_0^1 \int_0^1\varphi(v') \gamma(v')P_{I_2}(v|v')f_{I_2}(t, v') dv dv'\\
		&\phantom{=}+\lambda_\gamma\int_{0}^1\int_{0}^1\varphi(v')\gamma(v')P_{H_2}(v'|v)f_{H_2}(t,v)  dv dv'\\
		&\phantom{=} +\lambda_\gamma\int_{0}^1\int_0^1 (\varphi(v')P_{R}(v'|v)f_{R}(t,v)-\varphi(v)P_{R}(v|v')f_{R}(t,v')) dv  dv'  .
	\end{aligned}
	\label{eq:boltz.fR}
\end{align}
\end{itemize}
Equations \eqref{eq:boltz.fS1}--\eqref{eq:boltz.fR} have to hold for every $\varphi: V\to\R$. 

In order to obtain the equations for the macroscopic densities and viral load momentum of each compartment, we set $\varphi(v)=v^n$  in \eqref{eq:boltz.fS1}--\eqref{eq:boltz.fR}, with $n=0,1$, respectively. Since setting $\varphi(v)=v^n$ in the evolution equations \eqref{eq:boltz.fS1}--\eqref{eq:boltz.fR} leads to the appearance of the $(n+1)$--th moment of $f_i$, namely $\int_0^1f_i(t,v)v^{n+1} dv$, we need to find a closure. 
Specifically, for each  compartment we  consider a monokinetic closure in the form
\begin{equation}\label{monokin}
f_i(t,v)=\rho_i(t) \delta\left(v-\dfrac{n_i(t)}{\rho_i(t)}\right), \quad i \in \mathcal{X},
\end{equation}
i.e. we assume that all the agents of the same compartment at a given time $t$ have the same viral load. 
As already observed, if $\rho_i(t)=0$, then the mean viral  load is not well defined. Notwithstanding, since $f_i$ is defined as a Dirac delta, we have that if $\varphi(v)$ is a test function, then: 
$$
\int_0^1\varphi(v) f_i(t,v)  dv=\varphi\left(\dfrac{n_i(t)}{\rho_i(t)}\right)\rho_i(t).
$$
Then, we consider test functions such that 
\begin{equation}\label{eq:test_f_delta}
\varphi\left(\dfrac{n_i(t)}{\rho_i(t)}\right)\rho_i(t) \rightarrow 0, \quad \textrm{if} \,\, \rho_i(t) \rightarrow 0.
\end{equation} 
When $\varphi(v)=1$ or $\varphi(v)=v$, namely the test functions allowing to recover the masses and momentum, respectively, condition \eqref{eq:test_f_delta} is satisfied. Moreover, we  deal with terms in the form 
\begin{equation}\label{eq:term_clos}
\int_0^1\varphi(v)\psi(v) f_i(t,v)  dv
\end{equation} 
with $\psi=\eta,\,\gamma,\, \alpha_H$. Since we assumed that the probabilities $\eta,\,\gamma$ are constant, then the integral \eqref{eq:term_clos} with $\psi=\eta,\,\gamma$, is well defined, i.e. it is not divided by a vanishing density, for both test functions $\varphi(v)=1$ and $\varphi(v)=v$. As far as the probability of being isolated is concerned, i.e. $\psi(v)=\alpha_H(v)$, we have that, applying the monokinetic closure, the integral \eqref{eq:term_clos} reads 
\[
\varphi\left(\dfrac{n_{I_j}(t)}{\rho_{I_j}(t)}\right)\lambda_{H_j,I_j}(t)\alpha_H	\left(\dfrac{n_{I_j}(t)}{\rho_{I_j}(t)}\right)\rho_{I_j}(t), \quad j=1,2.
\]
Hence, we have to choose the isolation frequency and probability function in such a way that the latter quantity is well defined.

The macroscopic model is given by the following system of non--linear ordinary differential equations:
\begin{equation}
\begin{aligned}
\dot \rho_S &=  \lambda_b b -\lambda_\beta \nu_\beta\rho_S{\rho_I}  -\lambda_\mu \mu\rho_S \\
\dot \rho_{I_1} &= \lambda_\beta \nu_\beta\rho_S {\rho_I} -\lambda_{H_1,I_1}(t)\alpha_H\left(\dfrac{n_{I_1}}{\rho_{I_1}}\right)\rho_{I_1}-\lambda_\gamma\nu_1\rho_{I_1} -\lambda_\mu \mu \rho_{I_1}\\
\dot \rho_{I_2}&=\lambda_\gamma\nu_1\rho_{I_1}-\lambda_{H_2,I_2}(t)\alpha_H\left(\dfrac{n_{I_2}}{\rho_{I_2}}\right) \rho_{I_2}-\lambda_\gamma \nu_2\rho_{I_2}-\lambda_\mu \mu \rho_{I_2}\\
\dot \rho_{H_1}&= \lambda_{H_1,I_1}(t)\alpha_H\left(\dfrac{n_{I_1}}{\rho_{I_1}}\right)\rho_{I_1}-\lambda_\gamma\nu_1\rho_{H_1}-\lambda_d d \rho_{H_1}-\lambda_\mu \mu \rho_{H_1} \\
\dot \rho_{H_2}&=\lambda_\gamma\nu_1\rho_{H_1}+ \lambda_{H_2,I_2}(t)\alpha_H\left(\dfrac{n_{I_2}}{\rho_{I_2}}\right) \rho_{I_2}-\lambda_{\gamma}\nu_2\rho_{H_2} -\lambda_d d \rho_{H_2}-\lambda_\mu \mu \rho_{H_2} \\
\dot \rho_R &=\lambda_\gamma\nu_2\rho_{I_2}+\lambda_{\gamma}\nu_2\rho_{H_2}-\lambda_\mu \mu \rho_R\\
\dot n_{I_1}&= \lambda_\beta\nu_\beta v_0\rho_S{\rho_I}-\lambda_{H_1,I_1}(t)\alpha_H\left(\dfrac{n_{I_1}}{\rho_{I_1}}\right) n_{I_1}-\lambda_\gamma\nu_1 (n_{I_1}+(\nu_1-1)(\rho_{I_1}-n_{I_1}))-\lambda_\mu \mu n_{I_1} \\
\dot n_{I_2}&= \lambda_\gamma\nu_1(n_{I_1}+\nu_1 (\rho_{I_1}-n_{I_1}))-\lambda_{H_2,I_2}(t)\alpha_H\left(\dfrac{n_{I_2}}{\rho_{I_2}}\right) n_{I_2}-\lambda_\gamma\nu_2(2-\nu_2) n_{I_2}-\lambda_\mu \mu n_{I_2}  \\
\dot n_{H_1}&=\lambda_{H_1,I_1}(t)\alpha_H\left(\dfrac{n_{I_1}}{\rho_{I_1}}\right) n_{I_1}- \lambda_\gamma\nu_1(n_{H_1}+(\nu_1-1)(\rho_{H_1}-n_{H_1}))-\lambda_d d n_{H_1}-\lambda_\mu \mu n_{H_1}\\
\dot n_{H_2}&= \lambda_\gamma\nu_1(n_{H_1}+\nu_1 (\rho_{H_1}-n_{H_1}))+\lambda_{H_2,I_2}(t)\alpha_H\left(\dfrac{n_{I_2}}{\rho_{I_2}}\right) n_{I_2} -\lambda_\gamma \nu_2(2-\nu_2)n_{H_2}-\lambda_d d n_{H_2}-\lambda_\mu \mu n_{H_2} \\
\dot n_R&=\lambda_\gamma\nu_2 (1-\nu_2) n_{I_2}+\lambda_\gamma \nu_2(1-\nu_2)n_{H_2}-\lambda_\gamma \nu_2 n_R-\lambda_\mu \mu n_R.
\end{aligned}\label{macro_simplified}
\end{equation}
For convenience of notation, in (\ref{macro_simplified}) we have denoted with the upper dot the time derivative and omitted the explicit dependence on time of the state variables.

To model (\ref{macro_simplified}) we associate the following generic initial conditions
\begin{equation}\label{CI}
	\rho_S(0)=\rho_{S,0}>0,\,\,\rho_{i}(0)=\rho_{i,0}\geq 0,\,\,n_{i}(0)=n_{i,0}\geq 0,\quad i\in\{I_1,I_2,H_1,H_2,R\}.
\end{equation}
\textbf{Remark.} We observe that, by assuming $\gamma(v)=\nu_2$ and $\eta(v)=\nu_1$ (see Section \ref{Sec:micro}), we are not keeping into account the dependence of the transitions $(I_2\rightarrow R)$, $(H_2\rightarrow R)$, $(I_1\rightarrow I_2)$, $(H_1\rightarrow H_2)$ on the viral load value. However, with  these choices, the recovery rate of the infected individuals in classes $I_2$ and $H_2$  is $\lambda_{\gamma}\nu_2$, that is the decay rate of the viral load. Analogously, the rate of transition from $I_1$ [resp. $H_1$] to $I_2$  [resp. $H_2$]  is the increase rate of the viral load, $\lambda_{\gamma}\nu_1$.

\medskip 

As far as the products
$\lambda_{H_j,I_j}(t)\alpha_H(v)$, $j=1,2$,  are concerned, let us note that, if both the frequencies $\lambda_{H_j,I_j}$, $j=1,2$, and the probability $\alpha_H(v)$   are assumed to be constant, from (\ref{macro_simplified}) we retrieve a \textit{classical} SIR--like model with constant isolation rate. The qualitative analysis of the ensuing model  can be easily obtained and is here omitted.

We focus instead on the impact of viral load sensitivity of tests and frequency of testing activities  on the epidemic dynamics  
and consider the case that:
\begin{itemize}
		\item the probability for  infectious individuals to be isolated, $\alpha_H(v)$, linearly increases with their viral load: $\alpha_H(v)=\alpha v$, where    $\alpha\in[0,1]$ is a constant;
	\item the frequencies  $\lambda_{H_j,I_j}(t)$ are linearly dependent on the densities of infectious individuals: $\lambda_{H_j,I_j}(t)=\lambda_\alpha\rho_{I_j}(t)$, $j=1,2$, where    $\lambda_\alpha$ is a non--negative constant. Namely, we assume that the efforts made by public health authorities in screening and diagnostic activities increase with the increase of infectious presence in the community. Indeed, when infectious individuals are few, search activities could be highly expensive and little effective.
\end{itemize}
With these choices, in system (\ref{macro_simplified}), the isolation terms become
\begin{equation}
	\lambda_{H_j,I_j}(t)\alpha_H\left(\dfrac{n_{I_j}}{\rho_{I_j}}\right)=\lambda_\alpha\alpha n_{I_j},\quad j=1,2.\label{lalpha}
\end{equation} 
Equilibria and stability properties of model  (\ref{macro_simplified})--(\ref{lalpha}) will be investigated in the following section.

\section{Qualitative analysis}\label{sec:4}
Since in model (\ref{macro_simplified})--(\ref{lalpha}) the differential equations for $\rho_S,$ $\rho_{I_1},$ $\rho_{I_2},$ $n_{I_1},$ $n_{I_2}$  do not depend on $\rho_{H_1},$ $\rho_{H_2},$ $\rho_{R},$ $n_{H_1},$ $n_{H_2},$ $n_{R}$, it is not restrictive to limit our analysis to system
\begin{subequations}
	\begin{align}
		\dot \rho_S &=  \lambda_b b -\lambda_\beta \nu_\beta\rho_S(\rho_{I_1}+\rho_{I_2})  -\lambda_\mu \mu\rho_S\label{S'} \\
		\dot \rho_{I_1} &= \lambda_\beta\nu_\beta\rho_S (\rho_{I_1}+\rho_{I_2}) -\lambda_\alpha\alpha\rho_{I_1} n_{I_1}-\lambda_\gamma\nu_1\rho_{I_1} -\lambda_\mu \mu \rho_{I_1} \label{I1'}\\
		\dot \rho_{I_2}&=\lambda_\gamma\nu_1\rho_{I_1}-\lambda_\alpha\alpha \rho_{I_2}n_{I_2}-\lambda_\gamma \nu_2\rho_{I_2}-\lambda_\mu \mu \rho_{I_2}\label{I2'}\\
	\dot n_{I_1}&= \lambda_\beta\nu_\beta v_0\rho_S(\rho_{I_1}+\rho_{I_2})-\lambda_\alpha \alpha n_{I_1}^2-\lambda_\gamma\nu_1 (n_{I_1}+(\nu_1-1)(\rho_{I_1}-n_{I_1}))-\lambda_\mu \mu n_{I_1}  \label{n1'}\\
	\dot n_{I_2}&= \lambda_\gamma\nu_1(n_{I_1}+\nu_1 (\rho_{I_1}-n_{I_1}))-\lambda_\alpha \alpha n_{I_2}^2-\lambda_\gamma\nu_2(2-\nu_2) n_{I_2}-\lambda_\mu \mu n_{I_2}.   \label{n2'}
	\end{align}\label{macro2}
\end{subequations}
It is straightforward to verify that the region
\begin{equation}\label{regionD}
	\mathcal{D}=\left\{\left(\rho_S,\rho_{I_1},\rho_{I_2},n_{I_1},n_{I_2}\right)\in [0,1]^5 \,\Big|\, 0<\rho_S+\rho_{I_1}+\rho_{I_2}\leq \dfrac{\lambda_b b}{\lambda_\mu \mu},\,n_{I_1}\leq\rho_{I_1},\,n_{I_2}\leq\rho_{I_2}\right\}
\end{equation}
with initial conditions in (\ref{CI}) is positively invariant for model (\ref{macro2}), namely any solution of (\ref{macro2}) starting in $\mathcal{D}$ remains in $\mathcal{D}$ for all $t\geq 0$.

In the following, we search for model equilibria and derive suitable thresholds ruling their local or global stability.

\subsection{Disease--free equilibrium and its stability}
The model  (\ref{macro2}) has a unique disease--free equilibrium (DFE),  given by
\begin{equation*}
	DFE=\left(\dfrac{\lambda_b b}{\lambda_\mu \mu}, 0,0,0,0\right).
\end{equation*}
It is obtained
by setting the r.h.s. of equations  (\ref{macro2})  to zero and considering the case $\rho_{I_1}=\rho_{I_2}=0$.
\begin{prop}\label{ProplocalDFE}
	The DFE of system  (\ref{macro2}) is locally asymptotically  stable (LAS) if $\mathcal{R}_0<1$, where
	\begin{equation}\label{R0}
		\mathcal{R}_0=\lambda_\beta\nu_{\beta} \dfrac{\lambda_b b}{\lambda_\mu \mu}\dfrac{\lambda_\gamma\nu_1+\lambda_\gamma\nu_2   +\lambda_\mu \mu }{(\lambda_\gamma\nu_1   +\lambda_\mu \mu)(\lambda_\gamma\nu_2   +\lambda_\mu \mu)}.
	\end{equation}
	 Otherwise, if $\mathcal{R}_0>1$, it is unstable.
\end{prop}
\begin{proof}
The Jacobian matrix $J$ of system  (\ref{macro2}) evaluated at the DFE reads
\begin{equation*}
	J(DFE)=\left(\begin{array}{ccccc}	  -\lambda_\mu \mu & -\lambda_\beta\nu_{\beta} \dfrac{\lambda_b b}{\lambda_\mu \mu}   &		-\lambda_\beta\nu_{\beta} \dfrac{\lambda_b b}{\lambda_\mu \mu}  &		0   &		0\\
			0 &\lambda_\beta\nu_{\beta} \dfrac{\lambda_b b}{\lambda_\mu \mu} -\lambda_\gamma\nu_1   -\lambda_\mu \mu & \lambda_\beta\nu_{\beta} \dfrac{\lambda_b b}{\lambda_\mu \mu}  &0&0\\
			0&\lambda_\gamma\nu_1    &-\lambda_\gamma\nu_2-\lambda_\mu\mu&0&0\\
			0&	\lambda_\beta\nu_{\beta}v_0\dfrac{\lambda_b b}{\lambda_\mu \mu}+\lambda_\gamma \nu_1(1-\nu_1) & 	\lambda_\beta\nu_{\beta}v_0\dfrac{\lambda_b b}{\lambda_\mu \mu}&-\lambda_\gamma\nu_1(2-\nu_1) -\lambda_\mu\mu&0\\
			0&\lambda_\gamma \nu_1^2&0&\lambda_\gamma \nu_1(1-\nu_1)  &-\lambda_\gamma\nu_2(2-\nu_2) -\lambda_\mu\mu
		\end{array}\right).
\end{equation*}
One can immediately get the eigenvalues $l_1= -\lambda_\mu \mu<0$, $l_2=-\lambda_\gamma\nu_1(2-\nu_1) -\lambda_\mu\mu<0$, $l_3=-\lambda_\gamma\nu_2(2-\nu_2) -\lambda_\mu\mu<0$, while the other two are determined by the submatrix
$$\bar J=	\left(\begin{array}{cc}	
	\lambda_\beta\nu_{\beta} \dfrac{\lambda_b b}{\lambda_\mu \mu} -\lambda_\gamma\nu_1   -\lambda_\mu \mu & \lambda_\beta\nu_{\beta} \dfrac{\lambda_b b}{\lambda_\mu \mu}  \\
	\lambda_\gamma\nu_1    &-\lambda_\gamma\nu_2-\lambda_\mu\mu
\end{array}\right). $$
From the sign of the  entries of $\bar J$, it  follows that  det$(\bar J)\geq0$ implies tr$(\bar J)<0$. Hence, if det$(\bar J)>0$ or, equivalently, if $\mathcal{R}_0<1$, with $\mathcal{R}_0$ given in (\ref{R0}),
then the DFE is LAS. Otherwise, if $\mathcal{R}_0>1$,  it is unstable.\end{proof}

The threshold quantity $\mathcal{R}_0$ is the so--called \textit{basic reproduction number} for model (\ref{macro2}), a frequently used indicator for measuring the potential spread of an infectious disease in a community. Epidemiologically, it represents the average number of secondary cases produced by one primary infection over the course of the infectious period in a fully susceptible population. One can easily verify that the same quantity can be obtained  as the spectral radius of the so--called \textit{next generation} matrix 
\cite{vandendriessche2002}.

As far as the global stability of the DFE, we prove the following theorem.
\begin{theorem}\label{ThGAS}
	The DFE of system  (\ref{macro2}) is globally asymptotically stable  if $\mathcal{R}_0< 1$.
	\begin{proof}
Consider the following  function
\begin{equation*}
	\mathcal{L}=\dfrac{\lambda_\gamma\nu_1+\lambda_\gamma\nu_2   +\lambda_\mu \mu }{\lambda_\gamma\nu_1   +\lambda_\mu \mu}
	\rho_{I_1}+	\rho_{I_2}.
\end{equation*}
It is easily seen that the $\mathcal{L}$ is non--negative in $\mathcal{D}$ (see (\ref{regionD})) and also $\mathcal{L} = 0$ if and only if  $\rho_{I_1}=\rho_{I_2}=0$.
The time derivative of $\mathcal{L}$  along the solutions of system (\ref{macro2}) in $\mathcal{D}$ reads
\begin{align*}
	\dot{\mathcal{L}}=&\dfrac{\lambda_\gamma\nu_1+\lambda_\gamma\nu_2   +\lambda_\mu \mu }{\lambda_\gamma\nu_1   +\lambda_\mu \mu}
	\dot \rho_{I_1}+	\dot \rho_{I_2}\\
	=&\dfrac{\lambda_\gamma\nu_1+\lambda_\gamma\nu_2   +\lambda_\mu \mu }{\lambda_\gamma\nu_1   +\lambda_\mu \mu}
	\left(\lambda_\beta\nu_{\beta}\rho_S (\rho_{I_1}+\rho_{I_1})-\lambda_{\alpha}\alpha\rho_{I_1}n_{I_1}\right)-\lambda_{\alpha}\alpha\rho_{I_2}n_{I_2}-(\lambda_\gamma\nu_2   +\lambda_\mu \mu)(\rho_{I_1}+\rho_{I_1})\\
	\leq&(\lambda_\gamma\nu_2   +\lambda_\mu \mu)\left(\lambda_\beta\nu_{\beta}\rho_S\dfrac{\lambda_\gamma\nu_1+\lambda_\gamma\nu_2   +\lambda_\mu \mu }{(\lambda_\gamma\nu_1   +\lambda_\mu \mu)(\lambda_\gamma\nu_2   +\lambda_\mu \mu)}-1\right)(\rho_{I_1}+\rho_{I_1})\\
	\leq &(\lambda_\gamma\nu_2   +\lambda_\mu \mu)\left(\mathcal{R}_0 - 1\right)(\rho_{I_1}+\rho_{I_1}).
\end{align*}
It follows that $\dot{\mathcal{L}}\leq0$ for $\mathcal{R}_0 < 1$ with $\dot{\mathcal{L}}=0$ only if $\rho_{I_1}=\rho_{I_2}=0$.
Hence, ${\mathcal{L}}$ is a Lyapunov function on $\mathcal{D}$ and the largest compact invariant set in 
$\{\left(\rho_S,\rho_{I_1},\rho_{I_2},n_{I_1},n_{I_2}\right)\in \mathcal{D}: \dot{\mathcal{L}} = 0 \}$ is
the singleton \{DFE\}. Therefore, from the La Salle's invariance principle \cite{lasalle}, every solution to system (\ref{macro2}) with initial conditions in (\ref{CI})
approaches the DFE, as $t\rightarrow +\infty$.

As an alternative proof, one may adopt the approach developed by Castillo--Chavez \textit{et al.} in \cite{castillo}. 
\end{proof}\end{theorem}

\subsection{Endemic equilibria}
Let us denote with
$$EE=\left(\rho_S^E,\rho_{I_1}^E,\rho_{I_2}^E,n_{I_1}^E,n_{I_2}^E\right)$$
the generic endemic equilibrium of model (\ref{macro2}), obtained
by setting the r.h.s. of equations  (\ref{macro2})  to zero and considering the case $\rho_{I_1}+\rho_{I_2}>0$. 
Note that if it were $\rho_{I_1}^E=0$ [resp. $\rho_{I_2}^E=0$], from (\ref{I2'}) it would follow that  $\rho_{I_2}^E=0$  [resp. $\rho_{I_1}^E=0$].  Hence, it must be $\rho_{I_1}^E,\,\rho_{I_2}^E>0$.

More precisely, by rearranging equations (\ref{S'})--(\ref{I1'})--(\ref{I2'})--(\ref{n1'}), one obtains
\begin{equation}
\begin{aligned}
	\rho_S^E&=\dfrac{ \lambda _bb- ( \lambda _{\alpha }\alpha n_{I_1}^E + \lambda _{\gamma }\nu _1 +\lambda _{\mu }\mu ) \rho _{I_1}^E}{ \lambda _{\mu }\mu }\\
	\rho_{I_1}^E&=n_{I_1}^E\dfrac{ \lambda _{\alpha }\alpha  n_{I_1}^E+ \lambda _{\gamma }\nu _1(2-\nu_1) +\lambda _{\mu }\mu  }{ \lambda _{\gamma }\nu _1(1-\nu _1)+v_0( \lambda _{\alpha } \alpha  n_{I_1}^E+ \lambda _{\gamma }\nu _1 + \lambda _{\mu }\mu  )}\\
	\rho_{I_2}^E&=\dfrac{\lambda _bb -\lambda _{\beta } \nu _{\beta } \rho _S^E\rho _{I_1}^E -\lambda _{\mu }\mu   \rho _S^E}{\lambda _{\beta } \nu _{\beta } \rho _S^E}\\
	n_{I_2}^E&=\dfrac{ \lambda _{\gamma }\nu _1 \rho _{I_1}^E-(\lambda _{\gamma } \nu _2 + \lambda _{\mu }\mu)  \rho _{I_2}^E}{ \lambda _{\alpha }\alpha  \rho _{I_2}^E}.
\end{aligned}\label{EEcomp}
\end{equation}
Substituting expressions (\ref{EEcomp}) into (\ref{n2'}), one gets  $n_{I_1}^E$  as a positive root of the equation
\begin{equation}
 \lambda_\gamma\nu_1(n_{I_1}^E+\nu_1 (\rho_{I_1}^E-n_{I_1}^E))-\lambda_\alpha \alpha (n_{I_2}^E)^2-\lambda_\gamma\nu_2(2-\nu_2) n_{I_2}^E-\lambda_\mu \mu n_{I_2}^E=0  .	\label{n1EE}
\end{equation}	
Due to the complexity of equation (\ref{n1EE}), we renounce to get an explicit expression for
$n_1^E$ and, hence, to derive the existence conditions and number of endemic equilibria. However, we will make use of bifurcation analysis to show that a unique
branch corresponding to an unique endemic equilibrium emerges from the criticality, namely at DFE and $\mathcal{R}_0=1$.
	
\subsection{Central manifold analysis}\label{Sec: central manifold}

To derive a sufficient condition for the occurrence of a transcritical  bifurcation at $\mathcal{R}_0=1$, we can use
a bifurcation theory approach. We adopt the approach developed in
\cite{dushoff1998,vandendriessche2002}, which is based
on the general center manifold theory \cite{guckenheimer1983}. In
short, it establishes that the normal form representing the dynamics
of the system on the central manifold is, for $u$ sufficiently small, given by:
$$
\dot{u}=A{u}^{2}+B\lambda_\beta\nu_\beta{u},
$$
where
\begin{equation}
	A=\dfrac{\mathbf{z}}{2}\cdot D_\mathbf{{xx}}\mathbf{F}(DFE,\overline{\lambda_\beta\nu_\beta})\mathbf{w}^{2}\equiv\dfrac{1}{2}{\sum_{k,i,j=1}^{5} z_{k}w_{i}w_{j}\dfrac{\partial^{2}F_{k}(DFE,\overline{\lambda_\beta\nu_\beta})}{\partial x_{i}\partial x_{j}}} \label{eq:a}
\end{equation}
and
\begin{equation}
	B=\mathbf{z}\cdot D_{\mathbf{x}(\lambda_\beta\nu_\beta)}\mathbf{F}(DFE,\overline{\lambda_\beta\nu_\beta})\mathbf{w}\equiv{\sum^{5}_{k,i=1}}z_{k}w_{i}\dfrac{\partial^{2}F_{k}(DFE,\overline{\lambda_\beta\nu_\beta})}{\partial x_{i}\partial(\lambda_\beta\nu_\beta)}.\label{eq:b}
\end{equation}
Note that in (\ref{eq:a}) and (\ref{eq:b}) the product $\lambda_\beta\nu_\beta$ has been chosen
as bifurcation parameter, $\overline{\lambda_\beta\nu_\beta}$ is the critical value of $\lambda_\beta\nu_\beta$, $\mathbf{x}=\left(\rho_S,\rho_{I_1},\rho_{I_2},n_{I_1},n_{I_2}\right)$ is the state variables vector,
$\mathbf{F}$ is the right--hand side of system (\ref{macro2}),
and $\mathbf{z}$ and $\mathbf{w}$
denote, respectively, the left and right eigenvectors corresponding
to the null eigenvalue of the Jacobian matrix evaluated at criticality
(i.e. at DFE and $\lambda_\beta\nu_\beta=\overline{\lambda_\beta\nu_\beta}$).

Observe that $\mathcal{R}_0=1$ is equivalent to:
\[
	\lambda_\beta\nu_{\beta}=\overline{\lambda_\beta\nu_\beta}= \dfrac{\lambda_\mu \mu}{\lambda_b b}\dfrac{(\lambda_\gamma\nu_1   +\lambda_\mu \mu)(\lambda_\gamma\nu_2   +\lambda_\mu \mu) }{\lambda_\gamma\nu_1+\lambda_\gamma\nu_2   +\lambda_\mu \mu}
\]
so that the disease--free equilibrium is  stable if $\lambda_\beta\nu_\beta<\overline{\lambda_\beta\nu_\beta}$,
and it is unstable when $\lambda_\beta\nu_\beta>\overline{\lambda_\beta\nu_\beta}$. 

The direction of the bifurcation occurring at $\lambda_\beta\nu_\beta=\overline{\lambda_\beta\nu_\beta}$ can
be derived from the sign of coefficients (\ref{eq:a}) and (\ref{eq:b}).
More precisely, if $A>0$ [resp. $A<0$] and $B>0$, then at $\lambda_\beta\nu_\beta=\overline{\lambda_\beta\nu_\beta}$ there
is a backward [resp. forward] bifurcation.

For our model, we prove the following theorem.
\begin{theorem}
	System  (\ref{macro2}) exhibits a  forward bifurcation at DFE
	and $\mathcal{R}_0=1$.
	\begin{proof}
		From the proof of Proposition \ref{ProplocalDFE}, one can  verify that, when $\lambda_\beta\nu_\beta=\overline{\lambda_\beta\nu_\beta}$  (or, equivalently,
		when $\mathcal{R}_0=1$), the Jacobian matrix $J(DFE)$ admits a simple zero eigenvalue and the other eigenvalues
		have negative real part. Hence, the
		DFE is a non--hyperbolic equilibrium.\\ 
		It can be easily checked that a left and a right eigenvector associated
		with the zero eigenvalue so that $\mathbf{z\cdot}\mathbf{w}=1$ are:
		\begin{gather*}
			\mathbf{z}=\left(0,z_2,\dfrac{( \lambda _{\gamma }\nu _1 +\lambda _{\mu }\mu  )(\lambda _{\gamma }\nu _2 + \lambda _{\mu }\mu)}{\lambda _{\gamma }\nu _1 ( \lambda _{\gamma }\nu _1 +\lambda _{\mu }\mu)+(\lambda _{\gamma }\nu _2 + \lambda _{\mu }\mu)(\lambda _{\gamma }\nu _1+\lambda _{\gamma }\nu _2 + \lambda _{\mu }\mu) },0,0\right),\\
			\mathbf{w}=\left(-\dfrac{ \lambda _{\gamma }\nu _1 + \lambda _{\mu }\mu }{ \lambda _{\mu }\mu },1,\dfrac{  \lambda _{\gamma }\nu _1}{ \lambda _{\gamma }\nu _2+\lambda _{\mu }\mu},\dfrac{v_0(\lambda _{\gamma }\nu _1+\lambda _{\mu }\mu) +\lambda _{\gamma }\nu _1(1-\nu _1) }{ \lambda _{\gamma} \nu _1\left(2-\nu _1\right) +\lambda _{\mu }\mu  },w_5 \right)^{T},
		\end{gather*}
		with
		$$z_2=\dfrac{(\lambda _{\gamma }\nu _2 + \lambda _{\mu }\mu)(\lambda _{\gamma }\nu _1+\lambda _{\gamma }\nu _2 + \lambda _{\mu }\mu)}{\lambda _{\gamma }\nu _1 ( \lambda _{\gamma }\nu _1 +\lambda _{\mu }\mu)+(\lambda _{\gamma }\nu _2 + \lambda _{\mu }\mu)(\lambda _{\gamma }\nu _1+\lambda _{\gamma }\nu _2 + \lambda _{\mu }\mu) }$$
		and
		\[
		w_{5}=\lambda _{\gamma }\nu _1\dfrac{ \nu _1 \left(\lambda _{\gamma }+\lambda _{\mu }\mu  \right)+ v_0\left(1-\nu _1\right) \left(\lambda _{\gamma }\nu _1 +\lambda _{\mu }\mu  \right)}{\left( \lambda _{\gamma }\nu _1 \left(2-\nu _1\right)+ \lambda _{\mu }\mu \right) \left(\lambda _{\gamma }\nu _2 \left(2-\nu _2\right) + \lambda _{\mu }\mu \right)}.
		\]
		The coefficients $A$ and $B$ may be now explicitly computed. Considering
		only the non--zero components of the eigenvectors and computing the
		corresponding second derivative of $\mathbf{F}$, it follows that:
		\begin{align*}
			A&=z_{2}w_1\left[w_2\dfrac{\partial^{2}F_{2}(DFE,\overline{\lambda_\beta\nu_\beta})}{\partial \rho_S\partial \rho_{I_1}}+w_3\dfrac{\partial^{2}F_{2}(DFE,\overline{\lambda_\beta\nu_\beta})}{\partial \rho_S\partial \rho_{I_2}}\right]+z_2w_2w_4\dfrac{\partial^{2}F_{2}(DFE,\overline{\lambda_\beta\nu_\beta})}{\partial \rho_{I_1}\partial n_{I_1}}+z_3w_3w_5\dfrac{\partial^{2}F_{3}(DFE,\overline{\lambda_\beta\nu_\beta})}{\partial \rho_{I_2}\partial n_{I_2}}\\
			&=z_2w_1(1+w_3)\overline{\lambda_\beta\nu_\beta}-(z_2w_2w_4+z_3w_3w_5)\lambda_\alpha\alpha
		\end{align*}
		and
		\[
		B=z_{2}\left(w_{2}\dfrac{\partial^{2}F_{2}(DFE,\overline{\lambda_\beta\nu_\beta})}{\partial \rho_{I_1}\partial(\lambda_\beta\nu_\beta)}+w_{3}\dfrac{\partial^{2}F_{2}(DFE,\overline{\lambda_\beta\nu_\beta})}{\partial \rho_{I_2}\partial(\lambda_\beta\nu_\beta)}\right)=z_2(1+w_3)\overline{\lambda_\beta\nu_{\beta}} \dfrac{\lambda_b b}{\lambda_\mu \mu}
		\]
		where $z_2,\,z_3,\,w_2,\,w_3,\,w_4,\,w_5>0$ and $w_1<0$. Then, $A<0<B$. Namely, when $\lambda_\beta\nu_\beta-\overline{\lambda_\beta\nu_\beta}$ changes from negative to positive,
		DFE changes its stability from stable to unstable; correspondingly a negative unstable equilibrium becomes positive and locally asymptotically stable. This completes the proof.
\end{proof} \end{theorem}

\section{Numerical simulations}\label{sec:5}
In this section, we present and compare some numerical solutions of both the stochastic particle model \eqref{stoch_proc} and the macroscopic model (\ref{macro_simplified}). 

Our aim is to qualitatively  assess the interplay between the evolution of individuals' viral load and the disease spread and isolation control. Hence, demographic and epidemiological parameters values do not address a specific infectious disease and/or spatial area. They refer to a generic epidemic outbreak where control strategies rely on isolation of infectious individuals, as typically happens for new emerging infectious diseases (e.g., 2003--2004 SARS outbreak \cite{whoSARS}, 2014--2016 Western African Ebola virus epidemic \cite{cdcEbola}, the first phase of the ongoing COVID--19 pandemic \cite{whoCOVID}).

Numerical simulations are performed in MATLAB$^{\tiny{\textregistered}}$
 \cite{ma}.  We implement a Monte Carlo algorithm to simulate the  stochastic particle model \eqref{stoch_proc} and the 4th order  Runge--Kutta method with constant step size for integrating the system (\ref{macro_simplified}). Platform--integrated functions are used for getting the plots. 

\subsection{Parametrization}\label{Sec:par}
\begin{table}[t!]
	\centering\begin{tabular}{|@{}c|c|c@{}|}
		\hline
		Parameter& Description &Baseline value \\
		\hhline{|===|}
		    $\lambda_b$ &Frequency of new births or immigration & 1 days$^{-1}$\\	\hline
			$b$ & Newborns probability parameter &$2.58\cdot 10^{-5}$\\	\hline
			$\lambda_\mu$&  Frequency of natural deaths& 0.01 days$^{-1}$\\	\hline
			$\mu$&  Probability of dying of natural causes&$2.79\cdot 10^{-3}$\\	\hline
			$\lambda_\beta$ & Frequency of binary interactions& 1 days$^{-1}$\\	\hline
			$\nu_\beta$& Transmission probability parameter&0.29\\	\hline
			$v_0$&Initial viral load of infected individuals &0.01\\	\hline
			$\lambda_{H_1,I_1}(t)$&Frequency of isolation for $I_1$ members&See Section \ref{Sec:sim} \\	\hline
			$\lambda_{H_2,I_2}(t)$&Frequency of isolation for $I_2$ members&See Section \ref{Sec:sim} \\	\hline
			$\alpha_H(v)$& Probability for an infectious to be isolated &See Section \ref{Sec:sim} \\	\hline
			$\lambda_\gamma$& Frequency of viral load evolution&0.50 days$^{-1}$\\	\hline
			$\nu_1$&Factor of increase of the viral load&0.40\\	\hline
			$\nu_2$&Factor of decay of the viral load&0.20\\	\hline
			$\eta(v)$&Probability of having passed the viral load peak&$\nu_1$\\	\hline
			$\gamma(v)$&Probability of recovering&$\nu_2$\\	\hline
			$\lambda_d$& Frequency of disease--induced deaths& 0.01 days$^{-1}$\\	\hline
			$d$& Probability of dying from the disease& 0.10 \\
		\hline
	\end{tabular}\caption{List of  model parameters with corresponding description and baseline value.}\label{TabPar}
\end{table}

The time span of our numerical simulations is set to $t_f=1$ year
.
We are considering  an SIR--like model with demography and constant net inflow of susceptibles $\lambda_b b$. Since travel restrictions are usually implemented during epidemic outbreaks, we assume that $\lambda_b b$ accounts only for
new births (which can be assumed to be approximately constant due to the short time span of our analyses). Therefore, the net inflow of susceptibles is given by
	\begin{equation*}
		\lambda_b b=b_r\dfrac{\bar{N}}{N_{tot}},
	\end{equation*}
	where $b_r$ is the birth rate, $\bar N$ denotes the total resident population at the beginning of the epidemic, and 
	$N_{tot}$ is  the total system size. Note that $N_{tot}$ accounts for agents belonging to all model compartments $\mathcal{X}$ (including $B$, $D_\mu$, $D_d$), whereas $\bar N$ refers only to living individuals. 
	
	We assume a population of $\bar N=10^6$ individuals, representing, for example, the  inhabitants of a European metropolis.  Fluctuations in a time window of just over a year are considered negligible.
The most recent data by European Statistics  refer to  2019 and provide an average crude birth rate $b_r=9.5/1,000$ years$^{-1}$ \cite{euros1} and an average crude death rate $\lambda_\mu\mu=10.2/1,000$ years$^{-1}$ \cite{euros2}. 
The total (constant) system size $N_{tot}$ is set to $N_{tot}=\bar N/(1-\lambda_bbt_f)$, in such a way $N_{tot}=\bar N +\lambda_bbt_f N_{tot}$ is given by the sum of the initial population, $\bar N$, and the total inflow of individuals during the time span considered, $\lambda_bbt_f N_{tot}$. 

 For the epidemiological parameters we take the following baseline values:
 \begin{equation*}
\mathcal{R}_0=4,\,\,  \lambda_\gamma=1/2 \text{ days}^{-1},\,\, \nu_1=1/(5\lambda_\gamma),\,\, \nu_2=\nu_1/2,\,\,\lambda_d d= 9.997 \cdot 10^{-4}  \text{ days}^{-1}.
\end{equation*}
In particular, the product $\lambda_\gamma\nu_1$ can be interpreted as the inverse of the average time from exposure to viral load peak, whilst $\lambda_\gamma\nu_2$ as the inverse of the average time from viral load peak to recovery. The disease--induced death rate $\lambda_d d$ is estimated through the formula given by Day \cite{day2002}:
\begin{equation*}
	\label{d} \lambda_d d=(1-\lambda_\mu\mu T)\dfrac{C_F}{T},
\end{equation*}
where $C_F$ is the fatality rate and $T$ is the expected time from isolation until death. We assume $C_F=1$\% and $T=1/(\lambda_\gamma\nu_2)=10$ days.
As far as the  initial viral load  of infected individuals, $v_0$, is concerned, we assume that it is 1\% of the maximum reachable value ($v=1$), namely $v_0=0.01$. Finally, for the Monte Carlo simulation of the particle model \eqref{stoch_proc}, we further assume $\lambda_b=\lambda_{\beta}=1$ days$^{-1}$ and $\lambda_\mu=\lambda_d=0.01$ days$^{-1}$.

Initial data are set to the beginning of the epidemic, namely we consider  a single infectious individual in a totally susceptible population:
\begin{equation}\label{CIvalues}
	\rho_{S,0}=(\bar N-1)/N_{tot},\,\,\rho_{I_1,0}=1/N_{tot},\,\,n_{I_1,0}=v_0\rho_{I_1,0},\,\,\rho_{i,0}=n_{i,0}=0,\quad i\in\{I_2,H_1,H_2,R\}.
\end{equation}
All the parameters of the model as well as their baseline values are reported in Table \ref{TabPar}.

\subsection{The uncontrolled epidemic outbreak}\label{Sec:unc}
\begin{figure}[t]\centering
	\includegraphics[scale=1.1]{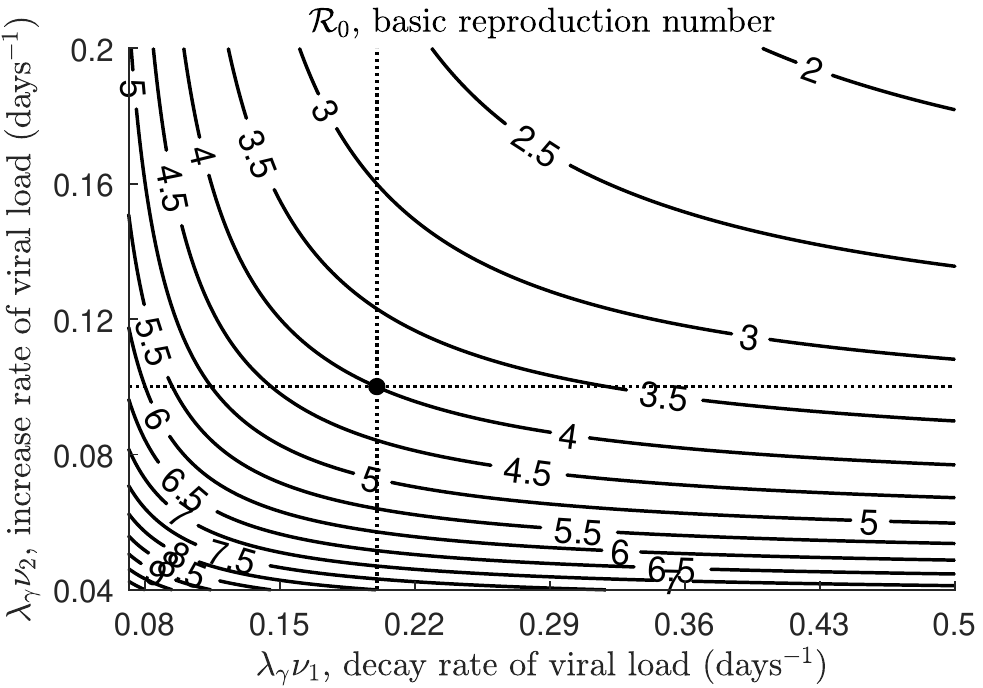}
	\caption{Contour plot of the basic reproduction number $\mathcal{R}_0$, as given in (\ref{R0}), versus the decay rate of viral load, $\lambda_\gamma\nu_1$,  and the increase rate of viral load, $\lambda_\gamma\nu_2$. Intersection between dotted black  lines indicates the value corresponding to the baseline scenario.  Other parameters values are given in Table \ref{TabPar}.  }\label{fig1}
\end{figure}

\begin{figure}[t]\centering
	\includegraphics[scale=1.1]{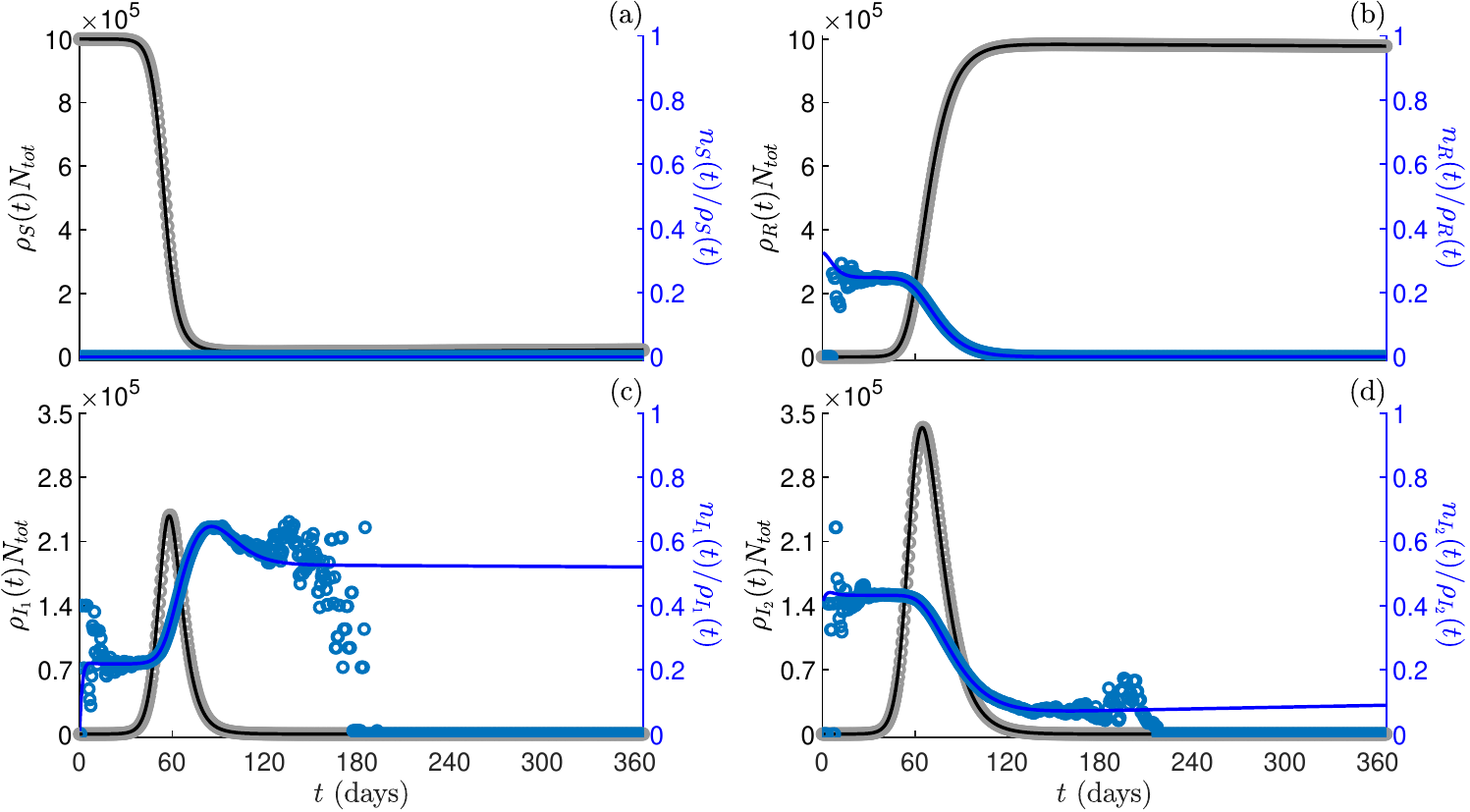}
	\caption{Epidemic dynamics in absence of isolation control ($\alpha_H\equiv 0$).   Compartment sizes (grey scale colour) and mean viral loads (blue scale colour) as predicted by model  (\ref{macro_simplified}) (solid lines) and by the stochastic process  \eqref{stoch_proc} (dots). Panel (a): susceptible, $S$. Panel (b): recovered, $R$. Panel (c): infectious with increasing viral load, $I_1$. Panel (d): infectious with decreasing viral load, $I_2$. Initial conditions and other parameters values are given in (\ref{CIvalues}) and Table \ref{TabPar}, respectively.  }\label{fig2}
\end{figure}

First, we numerically investigate the impact of the epidemiological parameters on the basic reproduction number $\mathcal{R}_0$, see (\ref{R0}). By considering the baseline parameters values,  we obtain that the ratio between  $\mathcal{R}_0$ and the transmission rate $\lambda_\beta\nu_\beta$ is about 13.83.  Fig. \ref{fig1} displays the contour plot of $\mathcal{R}_0$ versus  $\lambda_\gamma\nu_1$ (x--axis values)  and  $\lambda_\gamma\nu_2$ (y--axis values). We vary the average period of viral load increase, $1/(\lambda_\gamma\nu_1)$, in the range $[2,14]$ days and the average period of viral load decay, $1/(\lambda_\gamma\nu_2)$, in the range $[5,25]$ days. We obtain that $\mathcal{R}_0$ decreases with both $\lambda_\gamma\nu_1$ and  $\lambda_\gamma\nu_2$, from a maximum of $\mathcal{R}_0=10.39$ for $\lambda_\gamma\nu_1=1/14$ days$^{-1}$ and  $\lambda_\gamma\nu_1=1/25$ days$^{-1}$ to a minimum of $\mathcal{R}_0=1.87$ for $\lambda_\gamma\nu_1=1/2$ days$^{-1}$ and  $\lambda_\gamma\nu_2=1/5$ days$^{-1}$.

Let us now  set the basic reproduction number  to the baseline value $\mathcal{R}_0=4$ and investigate epidemic dynamics in absence of isolation control ($\alpha_H\equiv 0$). Numerical simulations are displayed in Fig. \ref{fig2}. We compare densities and viral loads mean. Specifically, solid lines refer to the solutions of the macroscopic model (\ref{macro_simplified}) and markers to those of the  stochastic particle model \eqref{stoch_proc}. We note a good match between the two approaches in predicting the dynamics of compartment sizes $\rho_iN_{tot}$, $i\in\mathcal{X}$ (grey scale colour): an epidemic outbreak invades the population, by reaching a prevalence peak of approximately $(\rho_{I_1}+\rho_{I_2})N_{tot}=531,000$ in  61 days; after 1 year the prevalence is almost zero and susceptible individuals are just about $21,700$. On the contrary, the dynamics of compartment mean viral loads  $n_i/\rho_i$, $i\in\mathcal{X}$ (blue scale colour) is different in the particle  w.r.t. the macroscopic model: the match is good as long as the corresponding compartment size is not so small  to make the effect of stochasticity relevant. For example, from Fig. \ref{fig2}(c)--(d), we see that, at the end of the epidemic wave, according to the particle model (blue markers)   the mean viral loads of infectious individuals fluctuate until approaching zero when the corresponding compartment becomes empty. Instead, the macroscopic model predicts that the same means remain approximately constant at a positive value (blue solid lines), suggesting that the  first moment $n_{I_1}$ [resp. $n_{I_2}$] and the density $\rho_{I_1}$ [resp. $\rho_{I_2}$] go to zero with the same \textit{speed}. This is due to the inconsistency of average quantities, like the mean viral loads, when the number of particles is very small. In that case, the deterministic macroscopic model cannot be justified by  means of the law of great numbers and statistical fluctuations must be taken into account.

\subsection{Viral load--dependent vs. constant isolation control}\label{Sec:sim}
\begin{figure}[t]\centering
	\includegraphics[scale=1.1]{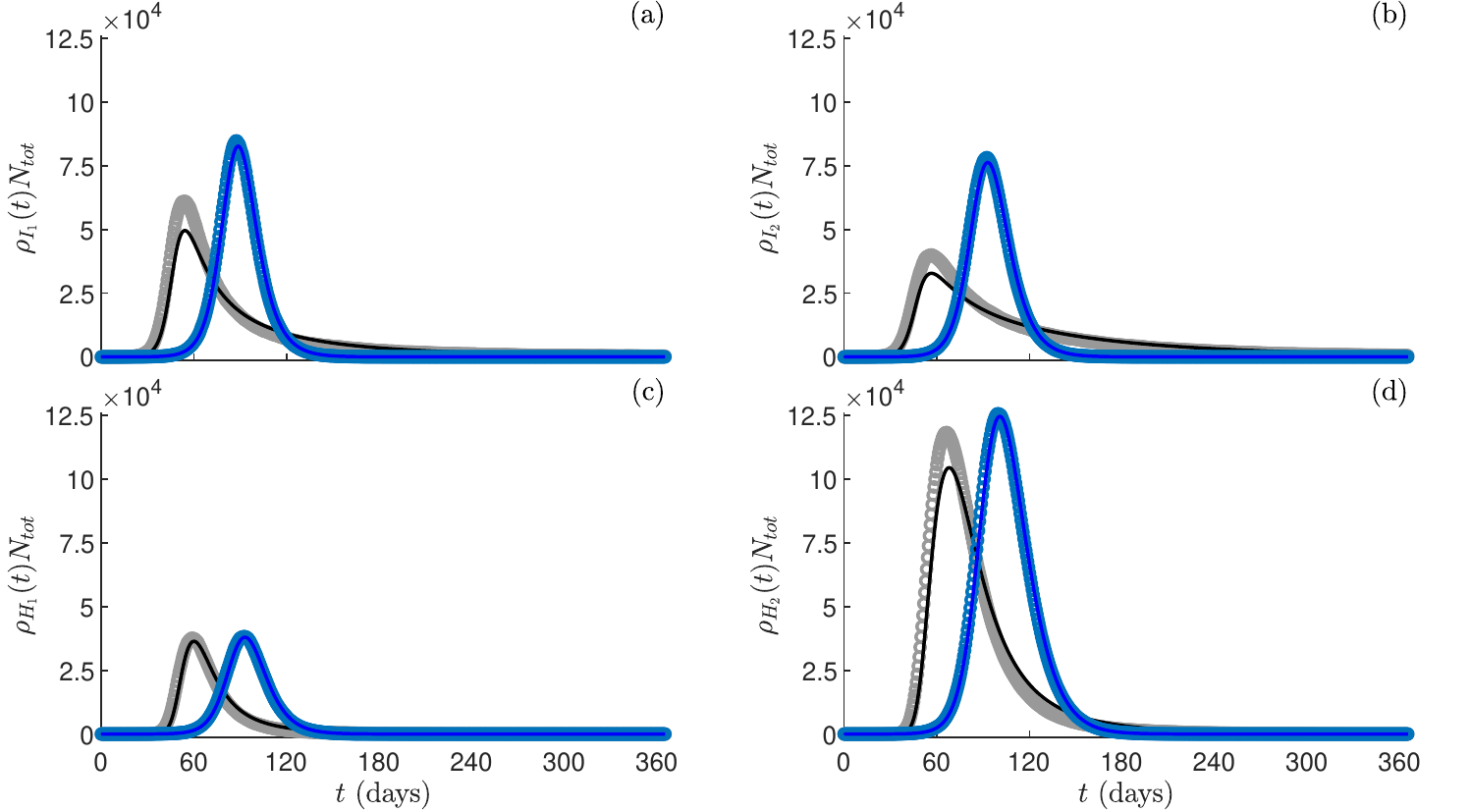}
	\caption{Viral load--dependent  vs. constant isolation control.   Numerical solutions as predicted by model  (\ref{macro_simplified}) (solid lines) and by the particle model  \eqref{stoch_proc} (markers) in scenarios \ref{S1} (grey scale colour) and \ref{S2} (blue scale colour).
		Panel (a): compartment size of infectious individuals with increasing viral load, $I_1$. Panel (b): compartment size of infectious individuals with decreasing viral load, $I_2$. Panel (c): compartment size of isolated individuals with increasing viral load, $H_1$. Panel (d): compartment size of isolated individuals with decreasing viral load, $H_2$.  Initial conditions and other parameters values are given in (\ref{CIvalues}) and Table \ref{TabPar}, respectively.  }\label{fig3}
\end{figure}

\begin{figure}[t]\centering
	\includegraphics[scale=1.1]{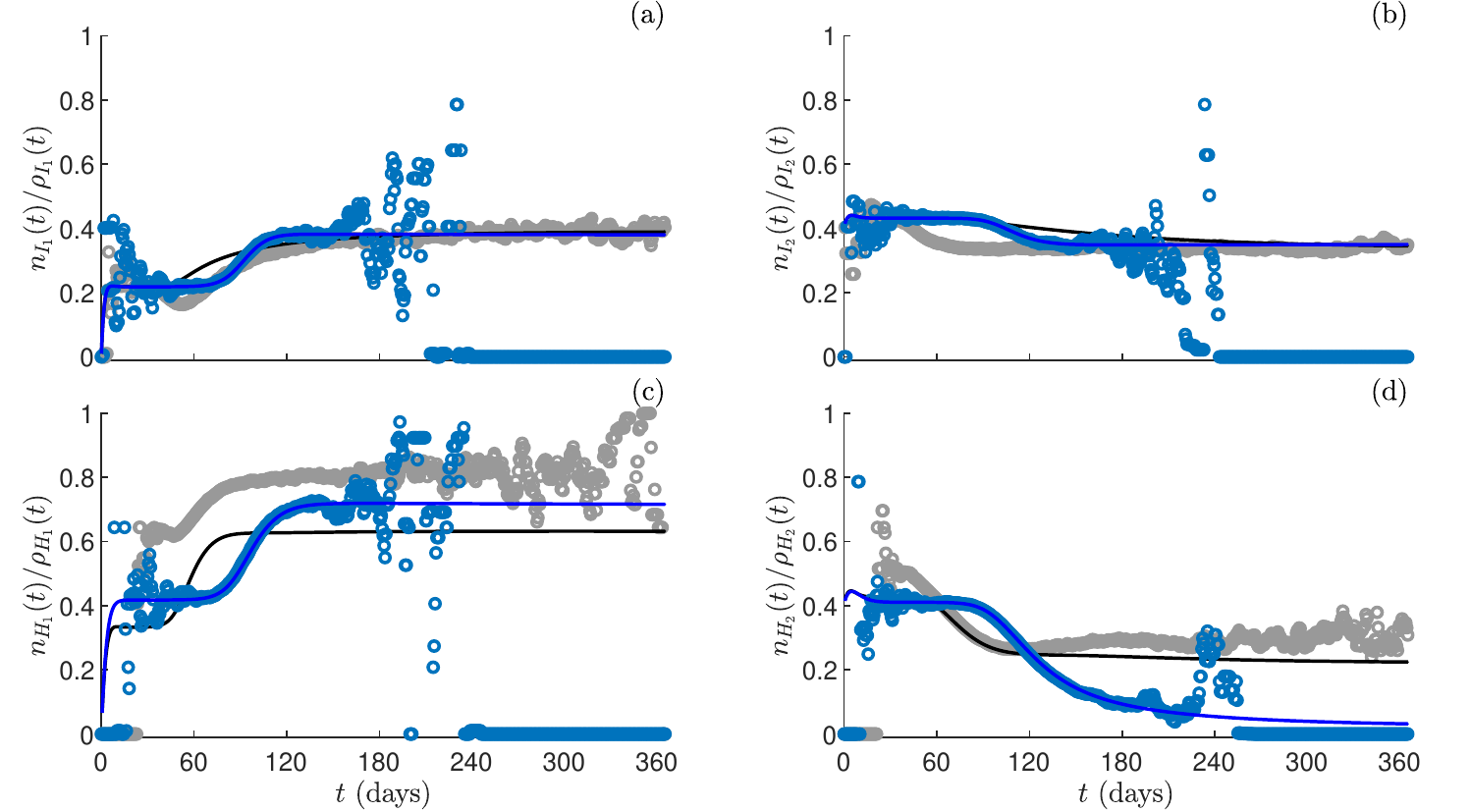}
	\caption{Viral load--dependent  vs. constant isolation control.   Numerical solutions as predicted by model  (\ref{macro_simplified}) (solid lines) and by the particle model  \eqref{stoch_proc} (markers) in scenarios \ref{S1} (grey scale colour) and \ref{S2} (blue scale colour). Panel (a): mean viral load of infectious individuals with increasing viral load, $I_1$. Panel (b): mean viral load  of infectious individuals with decreasing viral load, $I_2$. Panel (c): mean viral load  of isolated individuals with increasing viral load, $H_1$. Panel (d): mean viral load  of isolated individuals with decreasing viral load, $H_2$. Initial conditions and other parameters values are given in Sections \ref{Sec:par}, \ref{Sec:sim}.  }\label{fig4}
\end{figure}

Here, we introduce the isolation control in the epidemic model and assess how the frequency of testing and the viral load sensitivity of tests can affect epidemic dynamics. To this aim, we
compare the following simulation scenarios:
\begin{enumerate}[label=\textbf{S\arabic*},ref=S\arabic*]
	\item \label{S1} viral load--dependent isolation, as studied here: $\lambda_{H_j,I_j}(t)=\lambda_\alpha\rho_{I_j}(t)$, $j=1,2$, and $\alpha_H(v)=\alpha v$;
	\item \label{S2} constant isolation, as in \textit{classical} epidemic models: $\lambda_{H_j,I_j}(t)=\lambda_\alpha$, $j=1,2$, and $\alpha_H(v)=\alpha$.
\end{enumerate}
In order to make the two scenarios properly comparable, we make the following considerations. In the case \ref{S2}, the product $\lambda_\alpha\alpha$ represents the rate at which infectious individuals are isolated in the unit of time. In the case \ref{S1}, in the microscopic model,  the sane rate is given by $\lambda_\alpha\alpha$ multiplied by the individual microscopic viral load $v$; whereas, in the macroscopic model (\ref{macro_simplified}),  due to the monokinetic closure, this rate is given by $\lambda_\alpha\alpha$ multiplied by the momentum $n_{I_j}$, $j=1,2$, related to the corresponding compartment (see  (\ref{lalpha})). Thus, we assume that the value of $\lambda_\alpha\alpha$ in  scenario \ref{S1} (say, $\left.\lambda_\alpha\alpha\right.|_{S1}$) is given by the value adopted in  scenario \ref{S2} ($\left.\lambda_\alpha\alpha\right.|_{S2}$) rescaled by a normalization factor $M$:
\begin{equation*}
\left.\lambda_\alpha\alpha\right.|_{S1}=\dfrac{\left.\lambda_\alpha\alpha\right.|_{S2}}{M},
\end{equation*}
where $M$ represents an average quantity for the $n_{I_j}$'s, $j=1,2$. In order to estimate $M$, we consider model (\ref{macro_simplified})  in absence of isolation control ($\alpha_H\equiv 0$) and denote by $n_{I_1}^{unc}(t)$ and $n_{I_2}^{unc}(t)$ the corresponding solutions for $n_{I_1}(t)$ and $n_{I_2}(t)$, respectively. Then, $M$ is set to
\begin{equation*}
	M=\dfrac{\bar{n}_{I_1}+\bar{n}_{I_2}}{2},
\end{equation*}
where $\bar{n}_{I_j}$ are the average values of  $n_{I_j}^{unc}(t)$ over $[0,t_f]$, namely
$$\bar{n}_{I_j}=\dfrac{1}{t_f}\int_0^{t_f}n_{I_j}^{unc}(t)dt,\quad j=1,2.$$
 The numerical value for  $\left.\lambda_\alpha\alpha\right.|_{S2}$ is set to $\left.\lambda_\alpha\alpha\right.|_{S2}=0.1$ days$^{-1}$. For the Monte Carlo simulation we set $\lambda_\alpha=15$ days$^{-1}$ in scenario \ref{S1} and $\lambda_\alpha=1$ days$^{-1}$ in scenario \ref{S2}.
 
Figs. \ref{fig3} and \ref{fig4} display the numerical simulations in scenarios \ref{S1} (grey scale colour) and \ref{S2} (blue scale colour). Specifically, solid lines refer to the solutions of the macroscopic model (\ref{macro_simplified}) and markers to those of the stochastic particle model \eqref{stoch_proc}. As far as the match between the two approaches is concerned, we note that in the case of constant control \ref{S2}  considerations similar to those made in Section \ref{Sec:unc} apply.  Instead, in the case of viral load--dependent control  \ref{S1}, solutions by particle and macroscopic models are qualitatively similar but quantitatively different. This is an expected result because the derivation of the macroscopic model relies on an approximation through the monokinetic closure (\ref{monokin}), which acts by levelling  the viral loads of all agents belonging to a given class to their average value. However, notwithstanding the \textit{postulated} monokinetic closure, the matching is quite good, as the peak given by the macroscopic model is only mildly underestimated.

As far as the comparison between scenarios \ref{S1} and \ref{S2} is concerned, from Fig. \ref{fig3} we note that in the first case (grey scale colour) the epidemic outbreak occurs earlier and with a lower peak w.r.t. the second case  (blue scale colour), but the \textit{tails} of the infected curves are longer. In order to investigate these  differences more deeply, we consider solutions by the macroscopic model (\ref{macro_simplified}) and report in Table \ref{Tab} some relevant epidemiological quantities, including the value of  infectious prevalence peak and the time it occurs, and the endemic value of infectious prevalence, $(\rho_{I_1}^E+\rho_{I_2}^E)N_{tot}$. 
In scenario \ref{S1}, the endemic components $\left.\rho_{I_1}^E\right|_{S1}$, $\left.\rho_{I_2}^E\right|_{S1}$  are computed through the expressions in (\ref{EEcomp}). In particular, in our numerical set, equation (\ref{n1EE}) admits three positive roots, but just one of them makes all the other endemic components (\ref{EEcomp}) positive, hence a unique endemic equilibrium exists. In scenario \ref{S2}, one can easily verify that the unique endemic equilibrium has
	\begin{align*}
		\left.\rho_{I_1}^E\right|_{S2}&=\dfrac{\lambda _bb  \lambda _{\beta } \nu _{\beta } \left( \lambda _{\alpha }\alpha +\lambda _{\gamma }\nu _1+ \lambda _{\gamma }\nu _2+\lambda _{\mu }\mu  \right)-\lambda _{\mu }\mu   \left(\lambda _{\alpha }\alpha  + \lambda _{\gamma }\nu _1+\lambda _{\mu }\mu  \right) \left(\lambda _{\alpha }\alpha  +\lambda _{\gamma }\nu _2 +\lambda _{\mu }\mu  \right)}{\lambda _{\beta } \nu _{\beta }  \left(\lambda _{\alpha }\alpha  +\lambda _{\gamma }\nu _1 +\lambda _{\mu }\mu  \right)\left(\lambda _{\alpha }\alpha  +\lambda _{\gamma }\nu _1+ \lambda _{\gamma }\nu _2+\lambda _{\mu }\mu  \right)}\\
			\left.\rho_{I_2}^E\right|_{S2}&=\dfrac{\lambda _{\gamma }\nu _1}{ \lambda _{\alpha }\alpha  + \lambda _{\gamma }\nu _2+\lambda _{\mu }\mu  }\left.\rho_{I_1}^E\right|_{S2}.
\end{align*}
We also compute the value at the final time $t_f=1$ year of three cumulative quantities:  the cumulative incidence CI$(t)$, i.e. the total number of new cases in  $[0, t]$; the {cumulative isolated individuals}, i.e. the total number of infectious individuals that tested positive in  $[0, t]$,  and the cumulative deaths CD$(t)$, i.e. the disease--induced deaths in  $[0, t]$.
In our model we have, respectively:
\begin{align*}
	\text{CI}(t)&=N_{tot}\int_{0}^{t}\lambda_\beta\nu_\beta \rho_S(\tau)(\rho_{I_1}(\tau)+\rho_{I_2}(\tau))d\tau,
	\\
    \text{CH}(t)&=N_{tot}\int_{0}^{t}\left(\lambda_{H_1,I_1}(\tau)\alpha_H\left(\dfrac{n_{I_1}(\tau)}{\rho_{I_1}(\tau)}\right)\rho_{I_1}(\tau)+\lambda_{H_2,I_2}(\tau)\alpha_H\left(\dfrac{n_{I_2}(\tau)}{\rho_{I_2}(\tau)}\right)\rho_{I_2}(\tau)\right)d\tau,
    \\
	\text{CD}(t)&=N_{tot}\int_{0}^{t} \lambda_dd(\rho_{H_1}(\tau)+\rho_{H_2}(\tau))d\tau.
\end{align*}
\begin{table}[t!]
	\centering\begin{tabular}{|@{}c||c|c|c|c|c|c@{}|}
		\hline
		Scenario& $\max(\rho_{I_1}+\rho_{I_2})N_{tot}$ &arg$\max(\rho_{I_1}+\rho_{I_2})$ & CI$(t_f)$ &  CH$(t_f)$ & CD$(t_f)$&$(\rho_{I_1}^E+\rho_{I_2}^E)N_{tot}$ \\
		\hhline{|=======|}
		\ref{S1}  & $8.20 \cdot 10^4$ & 55.08 days &$7.87 \cdot 10^5$&$5.14\cdot 10^5$ &$6.33\cdot 10^3$&289.35\\ \hline
		\ref{S2}  & $15.60 \cdot 10^4$ &90.62 days&$7.70 \cdot 10^5$ &$5.13\cdot 10^5$&$6.34\cdot 10^3$&75.92\\
		\hline
	\end{tabular}\caption{Relevant quantities as predicted by model (\ref{macro_simplified}) in the case of viral load--dependent isolation \ref{S1} (first line) and in the case of constant isolation \ref{S2} (second line). First column: infectious prevalence peak, $\max(\rho_{I_1}+\rho_{I_2})N_{tot}$. Second column: time of infectious prevalence peak, arg$\max(\rho_{I_1}+\rho_{I_2})$. Third column: cumulative incidence at $t_f=1$ year,  CI$(t_f)$. Fourth column: cumulative isolated individuals at $t_f=1$ year,  CH$(t_f)$. Fifth column: cumulative deaths at $t_f=1$ year,  CD$(t_f)$. Sixth column: endemic infectious prevalence, $(\rho_{I_1}^E+\rho_{I_2}^E)N_{tot}$.  Initial conditions and other parameters values are given in (\ref{CIvalues}) and Table \ref{TabPar}, respectively.}\label{Tab}
\end{table}
From Table \ref{Tab}, we note that the epidemic peak in scenario \ref{S1}  is almost halved compared to the scenario \ref{S2} and occurs  36 days before. By contrast, the endemic infectious prevalence is
much greater in scenario \ref{S1} w.r.t. \ref{S2}: 289 vs. 76.  Interestingly, the differences in the cumulative quantities CI$(t_f)$, CH$(t_f)$, CD$(t_f)$ are, instead, minimal: in case of viral load--dependent isolation the cumulative incidence at 1 year is  approximately 2\% greater than the corresponding quantity in the case of constant isolation, while the cumulative isolated individuals [resp. deaths] are about 0.2\% greater  [resp. smaller].

The viral load--dependent isolation function reflects the assumption that an infectious individual with  high viral load  is more likely to be identified: it may represent the efficiency of the test that, according to its sensitivity, is capable of detecting different concentrations of virus particles per $ml$ \cite{larremore2021SA}. Assuming a constant isolation function means, instead, that all infectious individuals have the same probability of being detected and diagnosed. 
	In other words,  infectious individuals with sufficiently low [resp. high] viral load have a  probability of being diagnosed  higher [resp. lower] in scenario \ref{S2} with respect to \ref{S1}.

\subsection{Tracking individuals' viral load}
One of the advantages of  a particle  model is the possibility to track the trends of  all the agents of the system. Here, we are interested in tracking the evolution of individuals' viral load during the simulation time span. To this aim, we consider the particle model \eqref{stoch_proc} with viral load--dependent   isolation control (scenario \ref{S1}) and retrieve the viral load evolution of every single agent. In Fig. \ref{fig5}, we report the temporal dynamics of $v$ for five selected agents, who show different courses of the disease.  Different line markers and/or colours refer to the different epidemiological compartments the agents pass through; the meaning is specified in the figure legend. Note that two agents die after having acquired the infection: one of natural causes (first curve from the left), the other one from the disease (second curve from the left). The other three agents survive and finally recover from the infection: two of them are identified and isolated during the infectious period (third and fourth curve from the left), while the last one remains free to move  (fifth curve from the left).  We also remark that individuals may recover before their viral load becomes null and that it may take a long time after recovering  for $v$ to completely vanish. Such a trend  is linked to the choice of a constant probability of recovery, $\gamma=\nu_2$. This is in agreement with experimental observations of viral load curves, that show that individuals are no longer infectious before the complete disappearance of the virus \cite{cdcRubella,Goyal2020,He2020,flu}. Nonetheless, the mathematical description could be refined by setting a $v-$dependent and decreasing probability of recovering, $\gamma(v)$.

\begin{figure}[t]\centering
	\includegraphics[scale=1.1]{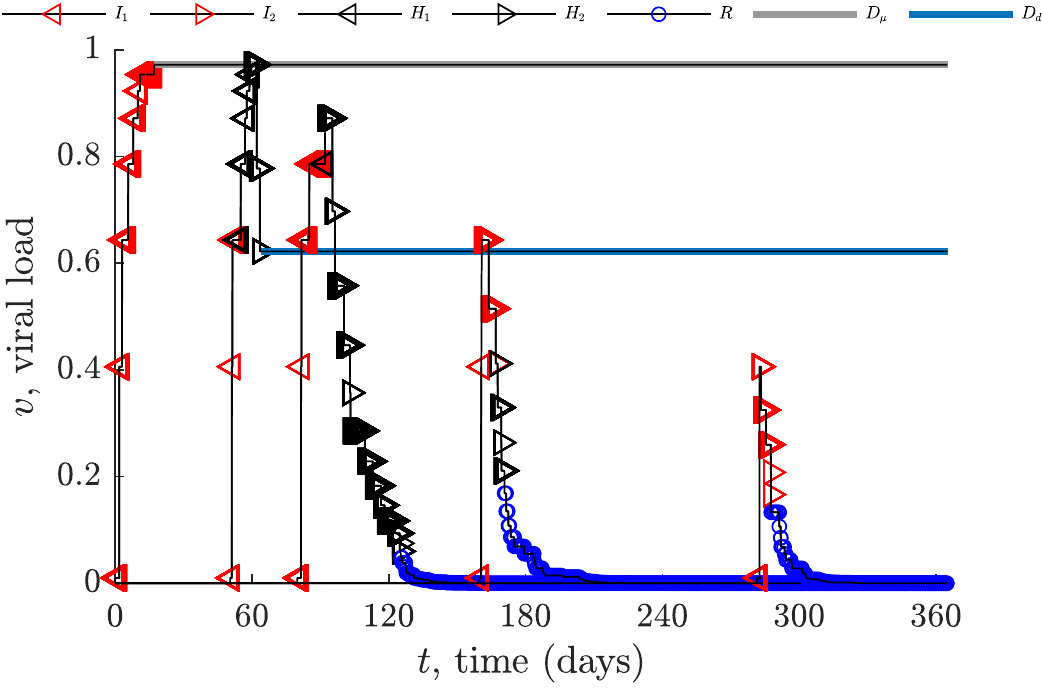}
	\caption{Viral load evolution from the time of infection exposure to the final time $t_f=1$ year for five system agents, as predicted by the stochastic particle model \eqref{stoch_proc} with viral load--dependent   isolation control (scenario \ref{S1}).   Different line markers and/or colours refer to the different epidemiological compartments the agent passes trough; the meaning is specified in the legend.  Initial conditions and other parameters values are given in (\ref{CIvalues}) and Table \ref{TabPar}, respectively. }\label{fig5}
\end{figure}

\section{Discussion and conclusion}\label{sec:6}
In this work, we have proposed a microscopic stochastic model allowing one to describe the spread of an infectious disease through social contacts. Each individual  is identified by the epidemiological compartment to which he/she belongs and by his/her viral load. Binary interactions between susceptible  and infectious individuals may cause the susceptible to acquire a positive viral load $v$ and, as a consequence, to get infected. The viral load progression due to  physiological processes is different from person to person, it determines the health status of the individuals and, therefore, the epidemiological compartment to which they belong. In particular,  we have here considered the case that the viral load influences explicitly the isolation mechanism, i.e. the switch from $I_i$ to $H_i$, $i=1,2$. In this sense, the present work manages to deal with the heterogeneity of the individuals' viral load, that is explicitly encoded in the individual viral load progression and in the probability of being diagnosed.

 We have derived from the stochastic particle model a kinetic description by means of evolution equations of the distribution of the viral load in each compartment.
  Finally, by making use of a monokinetic closure, we have obtained a hydrodynamic description. The ensuing macroscopic model is a system of non--linear ordinary differential equations for the macroscopic densities and viral load momentum of the compartments. 
We have performed a qualitative analysis allowing to state that our system has a unique disease--free equilibrium (DFE) that is globally asymptotically stable if $\mathcal{R}_0<1$ and that the system \eqref{macro_simplified} exhibits a transcritical forward bifurcation at  DFE and $\mathcal{R}_0=1$.

Our numerical tests have confirmed the matching between
the particle and the macroscopic models in the case in which the isolation function $\alpha_H$ is constant, thereby validating the macroscopic model as a reliable approximation of
the particle model more amenable to analytical investigations and quick and accurate numerical solutions. In the case of a viral load--dependent isolation, i.e. $\alpha_H(v)=\alpha v$, along with a density--dependent frequency, we have seen that the qualitative trend of the numerical solutions of the particle \eqref{stoch_proc} and macroscopic \eqref{macro_simplified} models are close but do not coincide exactly: this is a consequence of the approximation made through the monokinetic closure.
While the macroscopic model permits quicker numerical solutions, the particle model allows one to compute more accurately the viral load of the individuals: indeed, when a compartment is almost empty and there are few incoming and outgoing individuals, the statistical averages described by the macroscopic model are not reliable anymore. Moreover, the particle model allows one to track the viral load of every single agent and to investigate different possible individual evolutions of the disease.

Deliberately, we have not tried to match real scenarios by calibrating or comparing the results of our models with empirical data. 
In fact, our aim was first to propose a simple compartmental model including the viral load as microscopic variable. As a consequence, we wanted to explore prototypical scenarios and to compare them to those predicted by classical epidemic models, by focusing on the impact of having a viral load--dependent isolation in place of a constant isolation rate. We have seen that in the case of a viral load--dependent isolation the epidemic outbreak occurs earlier and with a lower peak (almost halved) w.r.t. to the constant isolation case. However, the cumulative disease--related quantities one year after the onset of the epidemic are comparable, while  the endemic infectious prevalence is much greater in the viral load--dependent isolation scenario. This may be explained in terms of the viral load sensitivity and frequency of the testing activities that are embodied in the choice of the functions $\alpha_H(v)$ and $\lambda_{H_1,I_1}(t)$, $\lambda_{H_2,I_2}(t)$, respectively.

 In the proposed framework, the description of the microscopic mechanisms and the heterogeneity of the viral load at the microscopic level allows one to derive a macroscopic model (more amenable, of course, to analytical and numerical investigations), that provides for a richer description of the  disease spreading in the host population. Here we only considered the explicit  influence of the viral load on the isolation mechanism, but, in principle, all the switches of individuals between compartments may depend on the viral load at the microscopic level, and on the viral load momentum at the macroscopic level.  Therefore, more complex situations, such as super--spreading events, that have been proved to be of the utmost importance for example during the  COVID--19 pandemic \cite{Goyal2020}, could be addressed. The heterogeneity of transmission could be included by making the disease transmission rate from infectious  to susceptible individuals dependent on the viral load.  Also, in such a  way, different initial viral loads of the  infectious individual first introduced in the community may give rise to different epidemic scenarios. 

\bigskip
\paragraph*{Acknowledgements} This work was supported by the Italian National Group for the Mathematical Physics (GNFM) of
National Institute for Advanced Mathematics (INdAM) through the grant `Progetto Giovani'. NL postdoctoral fellowship is funded by INdAM.
This work was also supported by Ministry of Education, University and Research through the MIUR grant Dipartimento di Eccellenza 2018-2022, Project no. E11G18000350001, and the Scientific Research Programmes of Relevant National Interest project n. 2017KL4EF3.

\bibliographystyle{abbrv}
\bibliography{DMrLnTa-SIR_viral_load}

\begin{thebibliography}{10}

\bibitem{Brauer}
F.~Brauer, P.~van~den Driessche, and J.~Wu.
\newblock {\em Mathematical Epidemiology}.
\newblock Springer, Berlin, 2008.

\bibitem{BBRDMCovid}
B.~Buonomo and R.~Della~Marca.
\newblock Effects of information--induced behavioural changes during the
  {COVID--19} lockdowns: the case of {I}taly.
\newblock {\em Royal Society Open Science}, 7(10):201635, 2020.

\bibitem{castillo}
C.~Castillo-Chavez, Z.~Feng, and W.~Huang.
\newblock On the computation of $\mathcal{R}_0$ and its role on global
  stability.
\newblock In {\em Mathematical Approaches for Emerging and Reemerging
  Infectious Diseases: An Introduction}. Springer, New York, 2002.

\bibitem{cdcEbola}
{CDC, Centers for Disease Control and Prevention}.
\newblock {2014--2016 Ebola outbreak in West Africa}.
\newblock
  \url{https://www.cdc.gov/vhf/ebola/history/2014-2016-outbreak/index.html},
  2016.
\newblock (Accessed on April 2021).

\bibitem{cdcRubella}
{CDC, Centers for Disease Control and Prevention}.
\newblock {Rubella-Laboratory Testing}.
\newblock \url{https://www.cdc.gov/rubella/lab/rna-detection.html}, 2020.
\newblock (Accessed on June 2021).

\bibitem{cevik2020virology}
M.~Cevik, K.~Kuppalli, J.~Kindrachuk, and M.~Peiris.
\newblock Virology, transmission, and pathogenesis of {SARS--CoV--2}.
\newblock {\em British Medical Journal}, 371:m3862, 2020.

\bibitem{day2002}
T.~Day.
\newblock On the evolution of virulence and the relationship between various
  measures of mortality.
\newblock {\em Proceedings of the Royal Society of London B},
  269(1498):1317--1323, 2002.

\bibitem{dimarco2020PRE}
G.~Dimarco, L.~Pareschi, G.~Toscani, and M.~Zanella.
\newblock Wealth distribution under the spread of infectious diseases.
\newblock {\em Physical Review E}, 102(2):022303, 2020.

\bibitem{dimarco2021PREPRINT}
G.~Dimarco, B.~Perthame, G.~Toscani, and M.~Zanella.
\newblock Kinetic models for epidemic dynamics with social heterogeneity.
\newblock Preprint: arXiv:2009.01140, 2021.

\bibitem{dushoff1998}
J.~Dushoff, W.~Huang, and C.~Castillo-Chavez.
\newblock Backwards bifurcations and catastrophe in simple models of fatal
  diseases.
\newblock {\em Journal of Mathematical Biology}, 36(3):227--248, 1998.

\bibitem{ecdc}
{ECDC, European Centre for Disease Prevention and Control}.
\newblock {Latest evidence on COVID--19 -- Infection}.
\newblock
  \url{https://www.ecdc.europa.eu/en/covid-19/latest-evidence/infection}, 2020.
\newblock (Accessed on June 2021).

\bibitem{euros2}
{European Commission - eurostat}.
\newblock {Deaths and crude death rate}.
\newblock
  \url{https://ec.europa.eu/eurostat/databrowser/view/tps00029/default/table?lang=en},
  2021.
\newblock (Accessed on April 2021).

\bibitem{euros1}
{European Commission - eurostat}.
\newblock {Live births and crude birth rate}.
\newblock
  \url{https://ec.europa.eu/eurostat/databrowser/view/TPS00204/bookmark/table?lang=en&bookmarkId=5b6e67ac-186d-4081-aa98-1453b77ec260},
  2021.
\newblock (Accessed on April 2021).

\bibitem{Fajnzylber2020}
J.~Fajnzylber, J.~Regan, K.~Coxen, H.~Corry, C.~Wong, A.~Rosenthal, D.~Worrall,
  F.~Giguel, A.~Piechocka-Trocha, C.~Atyeo, S.~Fischinger, A.~Chan, K.~T.
  Flaherty, K.~Hall, M.~Dougan, E.~T. Ryan, E.~Gillespie, R.~Chishti, Y.~Li,
  N.~Jilg, D.~Hanidziar, R.~M. Baron, L.~Baden, A.~M. Tsibris, K.~A. Armstrong,
  D.~R. Kuritzkes, G.~Alter, B.~D. Walker, X.~Yu, and J.~Z. Li.
\newblock {SARS--CoV--2} viral load is associated with increased disease
  severity and mortality.
\newblock {\em Nature Communications}, 11(1):5493, 2020.

\bibitem{Goyal2020}
A.~Goyal, D.~B. Reeves, E.~F. Cardozo-Ojeda, J.~T. Schiffer, and B.~T. Mayer.
\newblock Viral load and contact heterogeneity predict {SARS--CoV--2}
  transmission and super--spreading events.
\newblock {\em e{L}ife}, 10:e63537, 2020.

\bibitem{guckenheimer1983}
J.~Guckenheimer and P.~Holmes.
\newblock {\em Nonlinear Oscillations, Dynamical Systems, and Bifurcations of
  Vector Fields}.
\newblock Springer, Berlin, 1983.

\bibitem{He2020}
X.~He, E.~H.~Y. Lau, P.~Wu, X.~Deng, W.~Jian, X.~Hao, Y.~C. Lau, J.~Y. Wong,
  Y.~Guan, X.~Tan, X.~Mo, Y.~Chen, B.~Liao, W.~Chen, F.~Hu, Q.~Zhang, M.~Zhong,
  Y.~Wu, L.~Zhao, F.~Zhang, B.~J. Cowling, F.~Li, and G.~M. Leung.
\newblock Temporal dynamics in viral shedding and transmissibility of
  {COVID--19}.
\newblock {\em Nature Medicine}, 26:672--675, 2020.

\bibitem{SIR}
W.~Kermack and A.~G. McKendrick.
\newblock A contribution to the mathematical theory of epidemics.
\newblock {\em Proceedings of the Royal Society of London A}, 115:700–721,
  1927.

\bibitem{lasalle}
J.~La~Salle.
\newblock {\em Stability by Liapunov's Direct Method with Applications}.
\newblock Academic Press, New York--London, 1961.

\bibitem{larremore2021SA}
D.~B. Larremore, B.~Wilder, E.~Lester, S.~Shehata, J.~M. Burke, J.~A. Hay,
  M.~Tambe, M.~J. Mina, and R.~Parker.
\newblock Test sensitivity is secondary to frequency and turnaround time for
  {COVID}--19 screening.
\newblock {\em Science Advances}, 7(1):eabd5393, 2021.

\bibitem{flu}
N.~Lee, P.~K.~S. Chan, D.~S.~C. Hui, T.~H. Rainer, E.~Wong, K.-W. Choi,
  G.~C.~Y. Lui, B.~C.~K. Wong, R.~Y.~K. Wong, W.-Y. Lam, I.~M.~T. Chu, R.~W.~M.
  Lai, C.~S. Cockram, and J.~J.~Y. Sung.
\newblock Viral loads and duration of viral shedding in adult patients
  hospitalized with influenza.
\newblock {\em The Journal of Infectious Diseases}, 200(4):492--500, 2009.

\bibitem{loy2020KRM}
N.~Loy and L.~Preziosi.
\newblock Stability of a non--local kinetic model for cell migration with
  density dependent orientation bias.
\newblock {\em Kinetic \& Related Models}, 2020.
\newblock (To appear).

\bibitem{LnTa_NonCons}
N.~Loy and A.~Tosin.
\newblock Non--conservative {B}oltzmann--type kinetic equations for
  multi--agents systems with label switching.
\newblock Preprint: arXiv:2006.15550.

\bibitem{loyTosin2021PREPRINT}
N.~Loy and A.~Tosin.
\newblock A viral load--based model for epidemic spread on spatial networks.
\newblock Preprint: arXiv:2104.12107.

\bibitem{loy2020CMS}
N.~Loy and A.~Tosin.
\newblock Markov jump processes and collision--like models in the kinetic
  description of multi--agent systems.
\newblock {\em Communications in Mathematical Sciences}, 18(6):1539--1568,
  2020.

\bibitem{ma}
MATLAB.
\newblock {Matlab release 2020a. The MathWorks, Inc., Natick, MA}, 2020.

\bibitem{pareschi2013BOOK}
L.~Pareschi and G.~Toscani.
\newblock {\em Interacting {M}ultiagent {S}ystems: {K}inetic {E}quations and
  {M}onte {C}arlo {M}ethods}.
\newblock Oxford University Press, Oxford, 2013.

\bibitem{Simmonds}
M.~Simmonds, D.~Brown, and L.~Jin.
\newblock Measles viral load may reflect {SSPE} disease progression.
\newblock {\em Virology Journal}, 3:49, 2006.

\bibitem{vandendriessche2002}
P.~Van~den Driessche and J.~Watmough.
\newblock Reproduction numbers and sub--threshold endemic equilibria for
  compartmental models of disease transmission.
\newblock {\em Mathematical Biosciences}, 180(1):29--48, 2002.

\bibitem{spva}
Z.~Wang, C.~T. Bauch, S.~Bhattacharyya, A.~d'Onofrio, P.~Manfredi, M.~Perc,
  N.~Perra, M.~Salath{\'e}, and D.~Zhao.
\newblock Statistical physics of vaccination.
\newblock {\em Physics Reports}, 664:1--113, 2016.

\bibitem{whoSARS}
{WHO, World Helath Organization}.
\newblock {Severe acute respiratory syndrome (SARS)}.
\newblock
  \url{https://www.who.int/csr/don/archive/disease/severe_acute_respiratory_syndrome/en/},
  2004.
\newblock (Accessed on April 2021).

\bibitem{whotest}
{WHO, World Helath Organization}.
\newblock {Diagnostic testing for SARS--CoV--2. Interim guidance.}
\newblock
  \url{file:///C:/Users/rosde/AppData/Local/Temp/WHO-2019-nCoV-laboratory-2020.6-eng-1.pdf},
  2020.
\newblock (Accessed on May 2021).

\bibitem{whoCOVID}
{WHO, World Helath Organization}.
\newblock {Coronavirus disease (COVID--19) pandemic}.
\newblock
  \url{https://www.who.int/emergencies/diseases/novel-coronavirus-2019}, 2021.
\newblock (Accessed on April 2021).

\end{thebibliography}
\end{document}